%% file: mainAFT.tex
\PassOptionsToPackage{usenames, dvipsnames}{xcolor}

\documentclass[sigconf,nonacm, screen, authorversion, balance]{acmart}

\setcopyright{none} 
\acmConference[Anonymous Submission to ACM AFT 2022]{ACM Conference on Advances in Financial Technologies}{}{}
\acmYear{2020}

\input{packages}

\input{macros}

\begin{document}
	\title{\sname{}/\aname{}: Payment Channel Networks without Blockchain Synchrony} 

\author{O\u{g}uzhan Ersoy}
\email{oguzhan.ersoy@ru.nl}
\affiliation{%
	\institution{Radboud University\\Delft University of Technology}
}

\author{J\'er\'emie Decouchant}
\email{j.decouchant@tudelft.nl}
\affiliation{%
	\institution{Delft University of Technology}
}

\author{Satwik Prabhu Kumble}
\email{s.prabhukumble@tudelft.nl}
\affiliation{%
	\institution{Delft University of Technology}
}

\author{Stefanie Roos}
\email{s.roos@tudelft.nl}
\affiliation{%
	\institution{Delft University of Technology}
}

\begin{abstract}
	
	Payment channel networks (PCNs) enhance the scalability of block\-chains by allowing parties to conduct transactions off-chain, i.e, without broadcasting every transaction to all blockchain participants. To conduct transactions, a sender and a receiver can either establish a direct payment channel with a funding blockchain transaction or leverage existing channels in a multi-hop payment. 
	The security of PCNs usually relies on the synchrony of the underlying blockchain, i.e., 
	evidence of misbehavior needs to be published on the blockchain within a time limit.
	Alternative payment channel proposals that do not require blockchain synchrony rely on quorum certificates and use a committee to register the transactions of a channel. However, these proposals do not support multi-hop payments, a limitation we aim to overcome. 
	
	In this paper, we demonstrate that it is in fact impossible to design a multi-hop payment protocol with both network asynchrony and faulty channels, i.e., channels that may not correctly follow the protocol.   
	We then detail two committee-based multi-hop payment protocols that respectively assume synchronous communications and possibly faulty channels, or asynchronous communication and correct channels.  
	The first protocol relies on possibly faulty committees instead of the blockchain to resolve channel disputes, and enforces privacy properties within a synchronous network.
	The second one relies on committees that contain at most $f$ faulty members out of $3f{+}1$ and successively delegate to each other the role of eventually completing a multi-hop payment.  
	We show that both protocols satisfy the security requirements of a multi-hop payment and compare their communication complexity and latency. 
	\let\thefootnote\relax\footnotetext{
\fbox{\parbox{\dimexpr.95\columnwidth\fboxsep-2\fboxrule\relax}{This document is a preprint of a paper accepted at the ACM conference on Advances in Financial Technologies (AFT 2022).}}
}
\end{abstract}
	
\maketitle
	
\input{0_intro}

\input{1_relatedwork}

\input{2_impossibility}

\input{3_protocolnew}
\input{4_protocol_correctcommittees}

\input{5_security}
\input{6_conclusion}

\bibliographystyle{ACM-Reference-Format}
\bibliography{references}

\appendix
\input{7_committee}

\end{document}

%% file: packages.tex
\usepackage{amsmath,amsfonts,amsthm}
\usepackage{graphicx}
\usepackage{textcomp}
\usepackage{algpseudocode, algorithm}
\usepackage{tabularx,multirow}

\usepackage{epstopdf}

\usepackage[dvipsnames]{xcolor}

\usepackage{upgreek}
\usepackage{multirow}
\usepackage{array}

\usepackage{bbm}

\usepackage[caption=false]{subfig}

\usepackage{pifont}

\usepackage[normalem]{ulem}

\usepackage{xspace}

%% file: macros.tex
\newtheorem{theorem}{Theorem}
\newtheorem{lemma}[theorem]{Lemma}

\newtheorem{definition}[theorem]{Definition}

\newcommand{\sname}{\texttt{SyncPCN}\xspace}
\newcommand{\aname}{\texttt{PSyncPCN}\xspace}

\newcommand{\channel}{\gamma}

\newcommand{\mpath}{\mathsf{path}}
\newcommand{\mhp}{\mathsf{mhp}}

\newcommand{\comalgo}[1]{\texttt{\textcolor{Maroon}{/\!/ #1}}}

\newcommand{\adv}{\mathcal{A}}

\newcommand{\define}{:=}

\newcommand{\timebroadcast}{\delta}

\newcommand{\condpay}{\mathsf{CP}}

\newcommand{\descOfX}{\mathsf{desc}}

\newcommand{\itemOfX}{\mathsf{i}}

\newcommand{\cmark}{\ding{51}}%
\newcommand{\xmark}{\ding{55}}%

\newcommand{\cond}{\mathtt{cond}}

\newcommand{\unlocked}{\mathtt{unlocked}}
\newcommand{\locked}{\mathtt{locked}}
\newcommand{\paid}{\mathtt{paid}}
\newcommand{\revoked}{\mathtt{revoked}}

\newcommand{\pluseq}{\mathrel{{+}{=}}}
\newcommand{\minuseq}{\mathrel{{-}{=}}}

%% file: 0_intro.tex
\section{Introduction}
 
Payment channel networks (PCNs)~\cite{gudgeon2020sok} allow faster,  cheaper, and more energy-efficient transactions than blockchains.  Their key advantage is that most transactions can be conducted using only local communication whereas blockchains usually require broadcasting transactions to all participants for verification.  The largest PCN, Lightning~\cite{poon2016bitcoin},
has a capacity of more than 100 million euros as of January 29, 2022~\cite{1ml}.  

In a PCN,  two parties can open and fund bi-directional channels using a blockchain transaction. Once their channel has been opened, the two parties can conduct transactions by updating their shared account and moving the coins from one side to the other. 
Furthermore,  multi-hop payments allow forwarding a payment between a sender and receiver via a path of channels, enabling transactions between parties that have not opened a dedicated shared channel~\cite{DBLP:conf/sp/AumayrMEEFRHM21,DBLP:conf/sp/DziembowskiEFM19,DBLP:conf/eurocrypt/DziembowskiEFHH19,DBLP:conf/uss/AumayrMKM21,poon2016bitcoin,DBLP:conf/fc/0001BBKM19,DBLP:conf/ndss/MalavoltaMSKM19,lind2018teechain}.  
In a multi-hop payment, parties first lock coins of a channel with regard to a condition, then the payment is either executed by satisfying the condition or is revoked.

\subsection{Motivation}

The security of bi-directional payment channels relies on the underlying blockchain.
In case of a dispute about the channel balance, both parties can publish what they consider the most recent state on the blockchain, which can then determine the correct balance.  
For example, in a case where a malicious party publishes an old state of the channel balance, the honest party has to react to the published state. 
Within $\Delta$ blocks, the honest party can publish a penalty transaction~\cite{poon2016bitcoin} or raise a dispute~\cite{raiden} and claim the coins of the channel.
If the state of the honest party is not published in time --- which might be challenging in times where many transactions are waiting to be validated as well as in periods of network asynchrony --- she can lose her rightfully owned coins.  
Thus, blockchain synchrony is required to guarantee that disputes can be resolved and honest parties do not lose their coins.  

Second, for multi-hop payments in PCNs, parties lock coins in each channel for some time.
For example, in the Lightning Network, the locking time is typically chosen in the order of hours or days~\cite{mizrahi2020congestion}, and parties should complete the payment within this time period.  
Similar to payment channels, honest parties have to react to malicious updates within the time interval.  Thus,  there is a synchrony requirement that is not always satisfied, especially under attack~\cite{refBCsycn,refBCsycn2,harris2020flood}. 
For instance,  Flood\&Loot~\cite{harris2020flood} congests the underlying blockchain with channel closure requests containing invalid states. Due to the high number of transactions competing to be included in a block,  honest parties are unable to have the correct state included within a block in time and hence lose coins.   So, liveness faults of the blockchain layer, i.e.,  not being able to publish a transaction within time $\Delta$, cause safety faults,  i.e., loss of funds,  for payment channel networks.  

\subsection{Related Work}

Several works investigate asset transfers in asynchronous networks~\cite{gupta2016non,guerraoui2019consensus,collins2020online,auvolat2020money,naor2022payment}.
Concurrently and independently from our work, Naor and Keidar showed that establishing a payment channel with faulty parties on top of an asynchronous asset transfer system is impossible~\cite{naor2022payment}. More precisely, they demonstrate that only a unidirectional payment channel,  which in addition does not allow the payer in the channel to initiate closure, can be realized in this model.
They suspect that it might be possible to design a payment channel network using uni-directional channels where only the payee can close the channel with the use of HTLCs (hash-time lock contracts) from which the timeout component is removed. 
However, from the brief discussion given in the paper, it is unclear how to guarantee progress,  i.e.,  how to ensure that payments are eventually either successful or revoked,  without timeouts.  
There are three approaches that relate to blockchain asynchrony,  i.e.,  work for blockchains without the guarantee to have transactions included within time $\Delta$. First,  watchtowers~\cite{dryja2016unlinkable} allow parties to defer publication of the state to third parties.  In this manner,  
parties can become unresponsive, e.g., due to a denial-of-service attack, without being vulnerable to a malicious act by the other party in the channel. Yet, watchtowers are ineffective if the third parties are also not able to publish a transaction in time.   
In other words, watchtowers displace the responsibility of observing the blockchain but do not alleviate it.  
Second,  trusted execution environments (TEEs) can be leveraged to prevent the publication of invalid states, thus removing the need for a dispute period~\cite{lind2018teechain}.  Yet,  Intel SGX,  the TEE used in the existing designs,  is vulnerable to various side channel attacks~\cite{brasser2017software, wang2017leaky} and deprecated\footnote{\url{https://edc.intel.com/content/www/us/en/design/ipla/software-development-platforms/client/platforms/alder-lake-desktop/12th-generation-intel-core-processors-datasheet-volume-1-of-2/001/deprecated-technologies/}}. 
Third,  Brick~\cite{avarikioti2019brick} realizes asynchronous payment channels by involving a committee into the channel.  In exchange for fees, rational committee members keep track of the state of the channel and ensure that only valid states are published on the blockchain.  However,  Brick is only applicable for single-hop channels, not for PCNs.  

\subsection{Our Contributions}
We investigate multi-hop payments in PCNs with blockchain asynchrony.
By reducing the problem of a multi-hop payment to a fair exchange,  we show that multi-hop payments are only possible if the network is synchronous or if all channels in a multi-hop payment behave according to the protocol or in favor of honest parties, which can be enforced by a trusted third party. 

We design two protocols based on BFT committees: \sname and \aname.
\sname{} assumes synchronous communication and parties only need to trust the committees of their own channels, not the committees of other channels possibly involved in a multi-hop payment.   
\sname{} resembles the AMHL protocol, which is an improvement over Lightning in terms of privacy and security. Unlike AMHL, in \sname, each channel registers payments with a committee that enables dispute resolution among the parties without the need for blockchain synchrony. 

In contrast, \aname{} tolerates partially-synchronous communication in addition to blockchain asynchrony, but it requires that each committee involved in the payment acts honestly, i.e., that at most $f$ out of its $3f{+}1$ committee members are faulty.   We discuss how random peer sampling or a globally trusted committee can be used to obtain BFT committees with high probability.  
\aname{} relies on a total order broadcast protocol inside committees and on broadcasts between subsequent committees on a payment path.  Hence,  it is less efficient than \sname{}, which only relies on consistent broadcast.  
Given the classical FLP result~\cite{fischer1985impossibility}, \aname{} leverages existing total order broadcast protocols to provide termination if one assumes a partially synchrony~\cite{yin2019hotstuff, castro1999practical,decouchant2022damysus}, or probabilistic termination in full asynchrony~\cite{mostefaoui2015signature,kokoris2020asynchronous}. 
We prove the security of our protocols and provide a performance analysis.  While we assume honest and malicious parties during our analysis, our protocols can largely leverage incentives to ensure that rational parties behave honestly~\cite{avarikioti2019brick}.  

Our contributions are as follows:
\begin{itemize}
\item We prove that a secure multi-hop payment protocol with asynchronous communication and faulty channel updates is impossible. 
\item We design \sname{} and \aname{} for payment channel networks with blockchain asynchrony and synchronous or partially-synchronous communication, respectively.
\item We show that \sname{} and \aname{} enable secure multi-hop payments.
\item We derive communication and latency complexities for \sname{} and \aname{} and compare them with the existing protocols.
\end{itemize}
Our designs are the first payment channel networks that can be deployed without blockchain synchrony or trusted hardware.

The rest of this paper is organized as follows. 
\S\ref{sec:sota} provides the background knowledge.
\S\ref{sec:impossibility} proves that it is impossible to devise a protocol for asynchronous multi-hop payment channels without a trusted third party. 
\S\ref{sec:sync} and \S\ref{sec:async} respectively presents \sname{} and \aname{}.
\S\ref{sec:discussion} discusses incentives, the performance of \sname{} and \aname{} and compares them with existing payment channel networks. Finally, \S\ref{sec:conclusion} concludes this paper.

%% file: 1_relatedwork.tex
\section{Background}
\label{sec:sota}

This section introduces key concepts related to synchrony, committees, and communication.  
Furthermore,  we provide more details on payment channels, including a discussion of Brick, an asynchronous payment channel, and PCNs. 
We also formally define the security goals of multi-hop payments. 

\subsection{Time and Blockchain}

Nodes communicate by exchanging messages and are equipped with loosely synchronized clocks.  
Communication is called synchronous if there is a known upper bound $\delta$ on message transmission time, or asynchronous if there is no such upper bound. 
Messages are not lost but can be reordered. 

A blockchain is a distributed ledger of transactions that are grouped in blocks.
Time can be measured as a number of blocks appended to a blockchain. 
A secure blockchain should satisfy \textit{persistence} and \textit{liveness} properties~\cite{garay2015bitcoin,DBLP:conf/eurocrypt/PassSS17}.
A transaction is stable if it is unfeasible to remove it from the blockchain.
Persistence means that if a transaction is recorded as stable by an honest party, the rest of the honest parties do not record a conflicting transaction as stable.
Liveness implies that if an honest party wants to add a valid transaction into the blockchain, it will eventually be recorded as stable by all honest parties. 
Moreover, blockchain synchrony implies that a transaction is recorded as stable within $\Delta$ blocks.

\subsection{Byzantine Committees and Broadcast}

A committee is a set of nodes, each equipped with known signature keys, that jointly take or record decisions. They provide signed messages to show their agreement to a decision.  Nodes in the committee can be \emph{honest} (or \emph{correct}), which means that they strictly follow the protocol, or \emph{malicious} (or \emph{faulty}), which means that they might arbitrarily deviate from it. To tolerate at most $f$ faulty members, a Byzantine fault tolerant (BFT) committee contains at least $3f{+}1$ members.  
A committee that contains more than $f$ malicious nodes is said to be malicious.
A Byzantine quorum is a set of at least $2f{+}1$ replicas. A set of $2f{+}1$ signatures 
on a message is called a quorum certificate.
We rely on the following quorum properties. First, if nodes require a quorum certificate to update a value they store, then it is not possible to collect a quorum certificate 
on a different value. 
Second, when at least $f{+}1$ committee members declare that they have collected a quorum certificate for a given value, then the value returned is the one stored by correct committee members.  

To obtain guarantees regarding the messages delivered by correct committee members, one can leverage consistent, reliable, or total order broadcast.   
Informally, a \emph{consistent broadcast} used by an application guarantees that if a correct node broadcasts a message $m$, then all correct nodes eventually deliver $m$ to their application and that if two correct nodes deliver $m$ and $m'$, then $m{=}m'$~\cite{reiter1994secure}. 

\emph{Reliable broadcast}~\cite{bracha1985asynchronous,kozhaya2018rt} in addition requires that if a correct node delivers $m$ then all correct nodes eventually deliver it, regardless of whether the sender is correct or not. A \emph{total order (or atomic) broadcast} protocol additionally guarantees that correct nodes deliver messages in the same order~\cite{chandra1996weakest, kozhaya2021pistis}. Total order can be provided by a consensus protocol~\cite{castro1999practical, yin2019hotstuff,yu2019repucoin}.   

\subsection{Payment Channels and Brick}
\label{sec:pc_brick}

We consider bidirectional channels~\cite{decker2015fast}, as implemented in Lightning~\cite{poon2016bitcoin}.   
Two parties $P$ and $Q$ open a payment channel by publishing a funding transaction on the blockchain, through which they deposit $c_P$ and $c_Q$, respectively, coins for the channel.  
$(c_P,c_Q)$ denotes the initial \emph{balances} of the channel and the total amount of coins deposited, $c_P{+}c_Q$, is the \emph{capacity} of the channel.  
The balance is adjusted after each transaction, e.g., if $P$ sends $x$ coins to $Q$, then the balances will be $(c_P{-}x,c_Q{+}x)$ as long as $c_P {\geq} x$.  
Transactions are conducted \emph{off-chain}, i.e.,  the two parties locally sign a new state of the balance.
Parties can publish the final state on the blockchain to withdraw their coins with the latest balance.  

Brick~\cite{avarikioti2019brick} allows parties to establish a single-hop asynchronous payment channel. Each channel uses a BFT committee. State updates are sent to the committee and acknowledged by it. If a party wants to unilaterally close the channel, the committee confirms the last valid state and hence foregoes the need to publish on the blockchain within a certain time limit. Consistency between honest committee members is provided by a consistent broadcast protocol.  
Committee members only observe hashed payment channel states for privacy.
Brick does not require blockchain synchrony and can tolerate an asynchronous network.  

\subsection{Payment Channel Networks}
\label{sec:pcn}

The set of all payment channels form a network where the nodes are the parties and the edges are channels. Opening a channel requires a blockchain transaction and is hence costly and time-consuming.  
Instead, two parties that do not share a direct channel may send a \emph{multi-hop payment} along a path of payment channels~\cite{DBLP:conf/sp/AumayrMEEFRHM21,DBLP:conf/sp/DziembowskiEFM19,DBLP:conf/eurocrypt/DziembowskiEFHH19,DBLP:conf/uss/AumayrMKM21,poon2016bitcoin,DBLP:conf/fc/0001BBKM19,DBLP:conf/ndss/MalavoltaMSKM19} so that each party on the path pays their successor.

Let $\mhp$ be a multi-hop payment on path $path=(P_0,{\ldots},P_k)$ with corresponding channels $(\channel_0,{\ldots},\channel_{k-1})$.
Here, $path[0]\define P_0$ and $path[k] \define P_k$ denote the sender and the receiver respectively, and $path[1:k-1]$ are the intermediarty parties.
In \textit{source routing}-based multi-hop payment protocols, the sender $P_0$ chooses the path of the payment, and prepares a payload message $M_i$ for each party $P_i$. 
The payload message includes the necessary information regarding the payment, such as the amount, locking condition, the next party in the path, etc.

Initially, the state of each channel $\channel_i$ is $\unlocked$\footnote{In practice, there would be several simultaneous multi-hop payments, and these states are matched with the id of a multi-hop payment.
For simplicity, we ignore the corresponding ids, and focus on one multi-hop payment.}.
In the commit phase, each party $P_i$ locks the coins $v_i$ in $\channel_i$ (if they have enough balance and $f_i {=} v_{i-1} {-} v_{i}$, where $f_i$ is the fee of the intermediary $P_i$) with regard to a condition $\cond_i$, and then the channel state becomes $\locked$.  
The locked coins cannot be used by either party until the payment is finalized.  
In the payment execution phase, there are two possible outcomes for each locked channel: if the condition is satisfied and the payment is successful, then the final channel state becomes $\paid$; otherwise, the payment is revoked (canceled), and the final channel state is $\revoked$.

\subsection{Security and Privacy Definitions}
\label{sec:securepcn}

In general, the security of a multi-hop payment protocol is defined via \textit{balance security}, which often implicitly covers the \textit{correctness} and \textit{coin availability} properties as well.
Here, we define them separately.
Balance security implies that the total balance of an honest intermediary party does not decrease, and the sender should pay only if the receiver is paid~\cite{malavolta2017concurrency,DBLP:conf/uss/AumayrMKM21,eckey2020splitting}.
Correctness means that the payment is successful if all parties are honest and there is sufficient balance.  
The balance security definition covers \textit{safety} of the coins.
However, it does not take into account \textit{liveness}, i.e., the eventual transition from the initial to the final state.
Coin availability ensures that the payment is eventually finalized.

\begin{definition}[Balance Security]\label{def:bal}
	Let $\mhp$ be a multi-hop payment with $(v_0,{\ldots},v_{k-1})$ coins to be paid in channels of $(\channel_0,{\ldots},\channel_{k-1})$ such that $v_i{\geq} v_j$ for $0 {\leq} i {\leq} j {\leq} k{-}1$.
	If an honest intermediary party $P_i$ pays $v_{i+1}$ coins in $\channel_{i+1}$, $P_i$ receives the amount $v_{i}$ in $\channel_{i}$. If the sender $P_0$ pays $v_0$ coins in $\channel_{0}$, the receiver $P_k$ is paid $v_{k-1}$ coins in $\channel_{k-1}$. This implies that the final states of the channels $\channel_{i}$ and $\channel_{i+1}$ of an honest intermediary $P_i$ are identical, i.e., either both $\revoked$ or $\paid$.
\end{definition}

\begin{definition}[Correctness]\label{def:correct}
	Let $\mhp$ be a multi-hop payment with $(v_0,{\ldots},v_{k-1})$ coins to be paid in channels of $(\channel_0,{\ldots},\channel_{k-1})$ such that $v_i{\geq} v_j$ for $0 {\leq} i {\leq} j {\leq} k{-}1$.
	If the protocol is executed honestly and each channel $ \channel_{i}$ has a balance of at least $v_i$ and the locking conditions are satisfied, then the states of the channels of $\mhp.\mpath$  move to $\locked$ and then to $\paid$, otherwise, they remain $\unlocked$.
\end{definition}

\begin{definition}[Coin Availability]\label{def:availability}
	The channel of an honest party never stays forever in state $\locked$, i.e., it eventually transitions to $\paid$ or $\revoked$.
\end{definition}

A secure multi-hop payment protocol satisfies these three security notions. 
Instead of \textit{balance security}, \textit{atomicity} can be the main security goal. The difference between the two properties is highlighted by the \textit{wormhole attack}~\cite{DBLP:conf/ndss/MalavoltaMSKM19}, 
which is undesired, yet does not violate balance security.  
In this attack, the adversary steals the fees of other intermediaries by skipping them during a phase of the protocol. 
The attack violates atomicity,  which implies that if a channel revokes the payment all channels that precede this channel in the payment path also revoke, but not balance security.
We define atomicity by adding an additional requirement to the definition of balance security.

\begin{definition}[Atomicity]\label{def:atom}
	A multi-hop payment protocol satisfies atomicity if it satisfies balance security (Def.~\ref{def:bal}) and the following requirement: 
	For an honest sender $P_0$, if the condition $\cond_i$ for payment of channel $\channel_{i}$ is not satisfied for an honest party $P_{i+1}$, and the payment is $\revoked$, then there are no two channels $\channel_{\alpha}$ and $\channel_{\beta}$ such that the corresponding conditions $\cond_{\alpha}$ and $\cond_{\beta}$ are satisfied, and both channels are $\paid$ where $\alpha<i<\beta$.
\end{definition}

Regarding privacy, we consider the privacy concepts introduced in~\cite{moreno2017silentwhispers,malavolta2017concurrency}: 
\textit{Value privacy}, \textit{endpoint (sender or receiver) privacy} and \textit{relationship anonymity}. 
Value privacy means that the value of a payment is only revealed to the parties involved, i.e., sender, receiver, and intermediaries. 
Endpoint privacy states that the sender and receiver are not \emph{explicitly} revealed to any party that is not an endpoint.     
Last, relationship anonymity is a stronger notion of endpoint privacy stating that if there is an honest party among the intermediaries, then the corrupted intermediary parties cannot distinguish two payments sharing a segment of the path.  

We define the following games, which are formal variants of the ones defined in~\cite{malavolta2017concurrency,DBLP:conf/uss/AumayrMKM21}. 
In each game, an adversary tries to distinguish two different payments.
Unless stated otherwise, we assume that the payment values, fees, timelocks (if any) and path lengths of the two payments are the same (or they would be trivially distinguishable). 
Like in previous works, we do not consider side channel attacks~\cite{kumble2021lightning,nisslmueller2020toward}.
Also, the privacy properties are defined for the off-chain protocol and are not required to hold if the payment goes to the blockchain.

\noindent\textbf{Value Privacy Game}: Let an adversary $\adv \notin path$ choose two payment values $v^{0}$ and $v^{1}$ for a payment path $path$ where the channels in $path$ have sufficient capacities for both values.
Let $b\in\{0,1\}$ be chosen randomly. Let $\mhp^{b}$ be the corresponding multi-hop payment with payment value $v^{b}$.
In case of a successful payment of $\mhp^{b}$, $\adv$ wins the game by guessing the value of $b$:
\begin{eqnarray*}\label{eqn:v_priv}
Pr_{VP} \define Pr \left[  b'=b :  b' \leftarrow \adv^{v^{0},v^{1},path}, b \stackrel{R}{\leftarrow} \{0,1\} \right].
\end{eqnarray*} 

\begin{definition}[Value Privacy]\label{def:v_priv}
	We say that a multi-hop payment protocol satisfies value privacy if for every PPT (probabilistic polynomial-time) adversary $\adv$, the probability of winning Value Privacy Game is $Pr_{VP}=1/2+\epsilon$ where $\epsilon$ is negligible.
\end{definition}

\noindent\textbf{Endpoint Privacy Game}: Let an adversary $\adv \define P_i$ choose two payment paths $path^{0}$ and $path^{1}$ such that $(path^{0}[0],path^{0}[k]) \allowbreak \neq (path^{1}[0],path^{1}[k]) $ and $(P_{i-1},P_i,P_{i+1})=path^{0}[i-1:i+1]=path^{1}[i-1:i+1]$. 
Let $b\in\{0,1\}$ be chosen randomly. Let $\mhp^{b}$ be the corresponding multi-hop payment with path $path^{b}$, and $M^{b}_i$ be the payload message of $P_i$.
In successful payment of $\mhp^{b}$, $\adv$ wins the game by guessing the value of $b$:
\begin{eqnarray*}\label{eqn:ep_priv}
Pr_{EP} \define Pr \left[  b'=b :  b' \leftarrow \adv^{path^{0},path^{1},M^{b}_i}, b \stackrel{R}{\leftarrow} \{0,1\} \right].
\end{eqnarray*} 

\begin{definition}[Endpoint Privacy]\label{def:ep_priv}
	We say that a multi-hop payment protocol satisfies endpoint privacy if for every PPT adversary $\adv$, the probability of winning Endpoint Privacy Game is $Pr_{EP}=1/2+\epsilon$ where $\epsilon$ is negligible.
\end{definition}
 
\noindent\textbf{Relationship Anonymity Game}: Let an adversary $\adv \subset path[1:k-1]$ choose a path segment $path[1:k-1]$ such that there is at least one honest party in $path[1:k-1]$, and has two candidate senders $s^{0}$,$s^{1}$ and receivers $r^{0}$,$r^{1}$.
Let $b\in\{0,1\}$ be chosen randomly. If $b=0$, then $(path^{i}[0],path^{i}[k])= (s^{i},r^{i})$, otherwise $(path^{i}[0],path^{i}[k])= (s^{i},r^{1-i})$ for $i=0,1$.
Let $\mhp^{i}$ be the corresponding multi-hop payment with path $path^{i}$, and $\{M^{i}\}_\adv$ be the payload message(s) of $\adv$ for $i=0,1$.
In case of simultaneous successful payments of $\mhp^{0}$ and $\mhp^{1}$, $\adv$ wins the game by guessing the value of $b$:
\begin{eqnarray*}\label{eqn:r_anony}
Pr_{RA} \define Pr \left[  b'=b :  b' \leftarrow \adv^{path[1:k-1],\{M^{0}\}_\adv,\{M^{1}\}_\adv}, b \stackrel{R}{\leftarrow} \{0,1\} \right].
\end{eqnarray*}

\begin{definition}[Relationship Anonymity]\label{def:r_anony}
	We say that a multi-hop payment protocol satisfies relationship anonymity if for every PPT adversary $\adv$, the probability of winning the Relationship Anonymity Game is $Pr_{RA}=1/2+\epsilon$ where $\epsilon$ is negligible.
\end{definition}

%% file: 2_impossibility.tex
\section{Impossibility Proof for Multi-hop Payments}
\label{sec:impossibility}

We show that there is no secure multi-hop payment protocol on an asynchronous network with faulty channels. 
Here,  we refer to a channel as faulty if parties or committee members cause channel updates that are not in line with the protocol or refuse channel updates that are in line with the protocol.  We call a channel honest
if all updates happen in accordance with the protocol.  An honest channel requires that some but not all parties involved in the channel (parties or committees) behave in accordance with the protocol.  
Parties (or committees) follow the protocol either because they are honest or because they are rational and incentivized. 
For now,  we focus on honest parties and discuss incentives in Section~\ref{sec:discussion}.

\subsection{Fair Exchange}

A fair exchange protocol between parties $P$ and  $Q$ is defined as follows~\cite{AsokanSW98}.
Assume that $P$ has an item $i_P$ with the description $\descOfX_P$ and $Q$ has an item $\itemOfX_Q$ with the description $\descOfX_Q$.
At the end of the protocol, both parties should hold the item of the other party.

A fair exchange protocol has two phases: \textit{initialization} and \textit{claim-and-fund}.
During the initialization, parties agree on the description of the items, resulting in $P$ having $(\itemOfX_P,\descOfX_Q)$ and $Q$ having $(\itemOfX_Q,\descOfX_P)$.
In the claim-and-fund phase, $P$ obtains $\itemOfX_Q$ (wrt. $\descOfX_Q$) and $Q$ obtains $\itemOfX_P$ (wrt. $\descOfX_P$).

A fair exchange protocol should satisfy \textit{effectiveness}, \textit{fairness}, and \textit{timeliness}~\cite{AsokanSW98,pagnia1999impossibility}.
\begin{itemize}
	\item \textit{Effectiveness}: If no party misbehaves and the items match their descriptions, then $Q$ obtains $\itemOfX_P$ and $P$ obtains $\itemOfX_Q$; otherwise both parties abort.
	\item \textit{Fairness}: If honest $P$ does not obtain $\itemOfX_Q$, then $Q$ should not obtain $\itemOfX_P$, and vice versa.
	\item \textit{Timeliness}: Every honest party eventually terminates via either obtaining the correspoding item or aborting.
\end{itemize}

\subsection{The proof}

We explain the payment properties over the channels independently of their actual implementation.
The following lemma reduces a secure multi-hop payment to a fair exchange protocol.
The proof builds on previous work that showed that there exists no asynchronous cross-chain communication protocol with faulty nodes~\cite{zamyatin2019sok}.

\begin{lemma}\label{lem:red}
Assume $\prod_{SMHP}$ is a protocol that solves secure multi-hop payment. Then there exists a protocol $\prod_{FE}$ that solves fair exchange.
\end{lemma}
\begin{proof}
We first describe a fair exchange protocol $\prod_{FE}$ that uses a secure multi-hop payment protocol $\prod_{SMHP}$.
To realize a fair exchange between parties $P$ and $Q$, we consider two consecutive channels $\channel_{i}$ and $\channel_{i+1}$ of the multi-hop payment.
Assume that successful execution of $\prod_{SMHP}$ in $\channel_{i}$ assigns ownership of $\itemOfX_P$ to $Q$, and similarly $\channel_{i+1}$ assigns ownership of $\itemOfX_Q$ to $P$.
The descriptions of the exchanges can be defined in the corresponding conditions $\cond_i$ and $\cond_{i+1}$ and the existence of sufficient channel balances.
The exchange occurs if the state of the channel becomes $\paid$ and fails if the state becomes $\revoked$ or stays $\unlocked$.
We show that the secure multi-hop payment properties are equivalent to the fair exchange properties.

Effectiveness in fair exchange implies that if both parties $P$ and $Q$ honestly follow the protocol, and items match with the descriptions, then the exchange succeeds, otherwise, the exchange should fail for both parties.
In a secure multi-hop payment, correctness implies that if both channels $\channel_{i}$ and $\channel_{i+1}$ are honest, the channels have enough balance and the payment conditions are satisfied, then they will both accept the payment and update the state of the multi-hop payment with $\paid$, otherwise, they stay $\unlocked$.
Thus, from our description of a fair exchange protocol based on a secure multi-hop payment above, effectiveness and correctness are equivalent. 

Fairness in a fair exchange means that if an honest party does not obtain the item, then the other party does not as well. 
In a secure multi-hop payment, balance security implies that if the payment is $\paid$ in $\channel_{i+1}$, it should also be $\paid$ in  $\channel_{i}$, and similarly if the payment is $\revoked$ in $\channel_{i+1}$, then it should also be $\revoked$ in $\channel_{i}$.
In other words, either both channels become $\paid$ (successful), or both of them are $\revoked$ (aborted), which is equivalent to the fair exchange fairness.

Timeliness in fair exchange states that party $P$ (equiv. for $Q$) should eventually make the exchange or abort.
Coin availability for $\channel_{i}$ (and for $\channel_{i+1}$) states that if the channel is $\locked$, then it will eventually reach states $\paid$ or $\revoked$.
Also, note that if $\channel_{i}$ is not locked, then it is already $\unlocked$.
In other words, the state of the channel eventually reaches one of the final states or always stays in $\unlocked$.
Thus, coin availability and timeliness are equivalent.
\end{proof}

\begin{theorem}
There is no asynchronous secure multi-hop payment protocol that tolerates faulty channels without a trusted third party.
\end{theorem}

\begin{proof}
The proof follows from the reduction given in Lemma~\ref{lem:red} and the impossibility proof given in~\cite{pagnia1999impossibility}.
\end{proof}

%% file: 3_protocolnew.tex
\section{\sname: Blockchain-Asynchronous Payments on Synchronous Networks}
\label{sec:sync}

This section describes \sname, our multi-hop payment protocol that assumes synchronous communication. 
We utilize channels with a committee construction, which have also been used in Brick~\cite{avarikioti2019brick}. 
Other than the synchronous communication assumption, key differences to Brick are the introduction of timestamps for each committee member and update as well as the closing procedure. 
We assume that transactions to the blockchain are eventually included but not necessarily within time $\Delta$.  

\subsection{Model and Overview}
 
\paragraph{Threat Model} 

We assume a static adversary $\mathcal{A}$ that corrupts parties and committee members at the beginning of the protocol.
There is no upper limit on the number of corrupted parties, but the adversary cannot corrupt more than $f$ out of $N {\geq} 3f+1$ committee members of a channel of an honest party.
In other words, for an honest party $P_i$,  we assume that both channels $\channel_{i-1}$ and $\channel_{i}$ have at most $f$ malicious committee members among $N{\geq} 3f{+}1$.
This assumption should hold even if $P_{i-1}$ and $P_{i+1}$ are malicious. 
An honest party does not have to trust the committees of other parties for balance security and coin availability to hold. 

\paragraph{Payment channels} 
Each payment channel $(P,Q)$ involves a committee $W(P,Q)$ in the channel procedures: opening, update and closure.  
To establish a committee, two parties first agree on its members. 
During the channel opening, the parties publish the funding transaction, which deposits the coins and registers the committee. 
The coins can be redeemed from the channel if both parties sign the channel state, or if one party signs the state together with a Byzantine quorum of the committee.

To update the channel state, $P$ first computes the hash of the new channel state together with a random value,  signs the hash and shares the state, random value and the signature with $Q$.
Then, after $Q$ has validated the state, $Q$ also signs the same hash value and shares the signature with $P$.
Then, both parties send the signatures and the hash to the committee members via consistent broadcast.
The reason for not sharing the (unobfuscated) channel state with the committee members is the privacy of the individual channel updates.
In this way, the committee only sees the initial and final state of the channel.
The committee members validate the signatures on the hashed state.  If both signatures are valid, they acknowledge it by sending their signatures of the shared hash and register the state together with the time they have received it to the parties.
Once a party has received $2f+1$ signatures (necessary for the state to be available from any $f+1$ nodes if a channel has to be closed), she knows that the update is accepted.
An honest party does not participate in new update requests while there is an ongoing update or the channel closure has been initiated.
Note that the registration time is necessary for our protocol because of the time conditions.
The channel update takes $4 \timebroadcast$: exchanging signatures and states between the two parties takes $2 \timebroadcast$, informing the committee and receiving their responses are another $2 \timebroadcast$.

Closing a channel is a bit more complex than in single-hop asynchronous payment channels.
If both parties agree on the channel state, they can close the channel by signing the latest state, as in Brick. 
Otherwise, the party wanting to close, say $P$,  initiates the closure by sharing the latest state and random value used in the hash
with $W(P,Q)$.   $W(P,Q)$ checks that it is indeed the latest state by computing the hash and checking that it is equal to their latest hash.  Then,  $W(P,Q)$ informs $Q$ of the closure initiation.
If there are ongoing multi-hop payments, the committee waits until the timeouts of the multi-hop payments before closing the channel.
Note that these timeouts are computed with respect to the registration time of (the hash of) the channel state corresponding to the payment\footnote{The registration time of each committee member can be different. For the security of our protocol, we only require that an update sent by an honest party is received by honest committee members within $\delta$ time.}.
Once all the channel outputs are resolved, then the committee signs the final channel state accordingly.
$P$ can close the channel by publishing their signature together with the quorum certificate of the committee.

\paragraph{Overview of the protocol} 
Figure~\ref{fig:3hops} presents a payment example that illustrates \sname for four parties. 
Let us generalize the example for  $k{+}1$ parties where $P_0$ is the payment sender, $P_k$ is the payment receiver, and parties $P_1,\ldots,P_{k-1}$ are the intermediaries. 
The payment goes through $k$ channels: $\channel_{0}, \channel_{1},\ldots, \channel_{k-1}$ where $\channel_{i}$ denotes the channels between $P_{i}$ and $P_{i+1}$.  
In a successful multi-hop payment, in each channel $\channel_{i}$, $P_i$ pays $v_i \define v + \sum_{j=i+1}^{k-1} f_j$ coins to $P_{i+1}$ where $v$ is the value of the payment agreed by the sender and receiver, and $f_j$ denotes the fee of an intermediary $P_j$.
Here, a mechanism is needed to ensure balance security, i.e., to guarantee that honest $P_i$  receives $v_{i+1}$ coins in $\channel_{i+1}$ if she pays $v_{i}$ coins  in $\channel_{i}$.
Lightning achieves balance security by using HTLCs (hash-time lock contracts). 
However, because of the privacy concerns as well as the \textit{wormhole attack} against HTLC,  we adopt the AMHL (anonymous multi-hop lock) protocol~\cite{DBLP:conf/ndss/MalavoltaMSKM19}.

Here, we briefly explain the concept of \textit{conditional payments}  that we adopt from AMHL.  
A conditional payment $\condpay$ between payer $P$ and payee $Q$ can be defined with a tuple of $(v,T,Y)$ where $v$ is the payment value, $T$ is the timelock~\footnote{In AMHL, the time unit is defined over the blocks. In our protocol, we use the global time that is available to parties and committee members.} and $Y$ is the locking condition.
The conditional payment works as follows: First, $P$ locks $v$ coins with respect to the timelock $T$ and condition $Y$ with regard to an additively homomorphic one-way function $\mathcal{H}$~\cite{DBLP:conf/ndss/MalavoltaMSKM19}. 
$Q$ then can claim the amount $v$ by providing a witness (preimage) $y$ satisfying the condition $Y$, i.e., $Y=\mathcal{H}(y)$.
If the secret is not provided before time $T$, $P$ re-claims the amount.

\begin{figure}[t]
	\centering
	\includegraphics[width=\columnwidth]{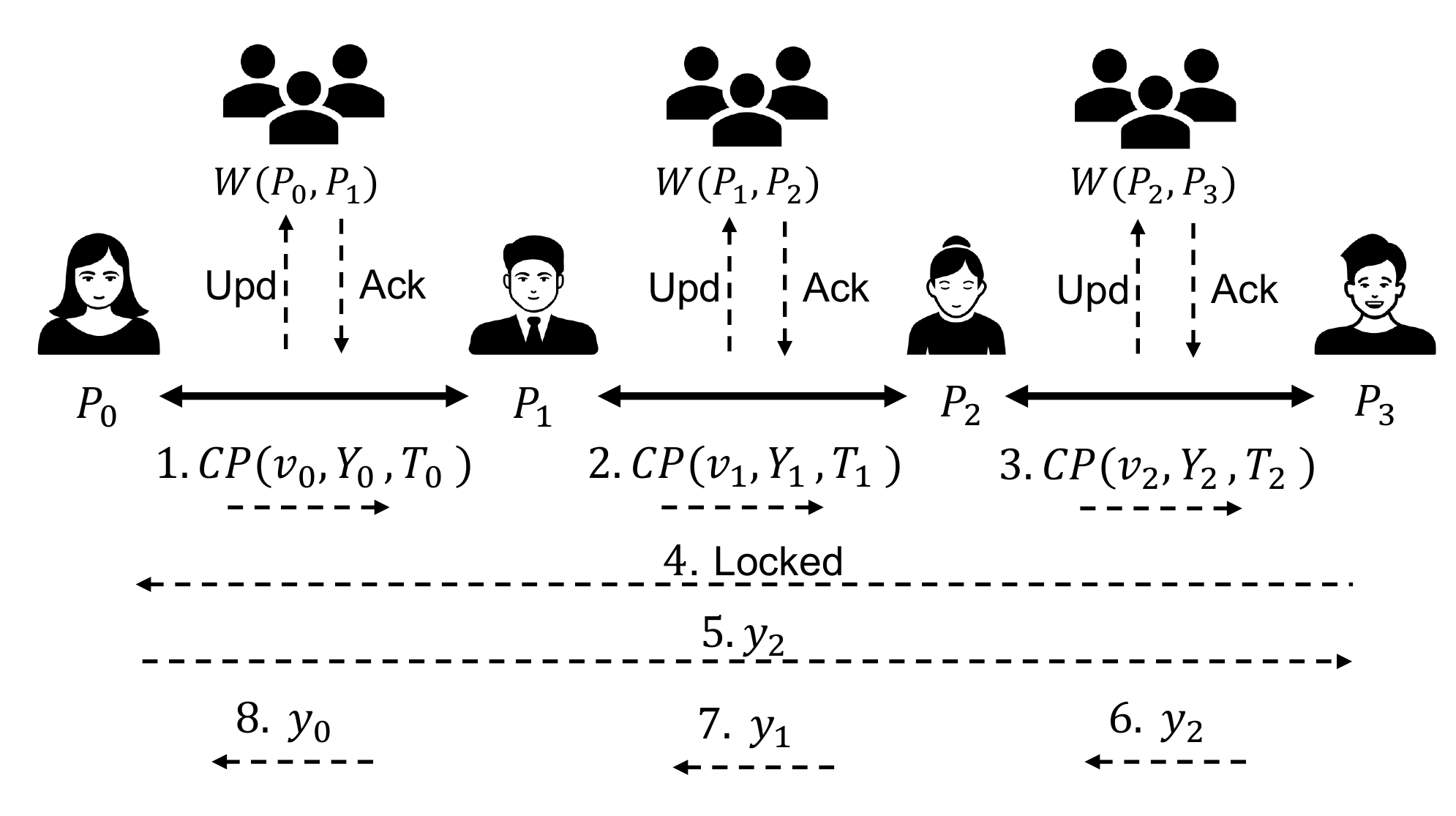}
	\caption{Illustration of the steps of \sname.}
	\label{fig:3hops}
\end{figure}

\subsection{Multi-hop Payment Protocol}

We assume that the sender and receiver agree on the payment value $v$ and the path is chosen by the sender.
The criteria on how to select the path are orthogonal to our protocol and any source routing algorithm can be used for selecting the path.
To improve privacy, we use onion routing with the Sphinx package format~\cite{danezis2009sphinx} for the communication between parties in the path.
The layered encryption allows intermediaries to learn the predecessor and successor on the path as well as a per-hop payload that includes payment information necessary such as the payment value or timeouts\footnote{As described in \url{https://github.com/lightningnetwork/lightning-rfc/blob/master/04-onion-routing.md}.}.  

Now, we explain our protocol in three steps: First, the sender creates the locking conditions for each channel in the path wrt. the AMHL protocol. 
Second, from sender to receiver, parties lock the coins if they accept the conditions and there is enough balance.
In the final step, the payment is executed: either the condition is satisfied and payment is completed or it is revoked.

\paragraph{Setup} First, the sender $P_0$ chooses random values  $\ell_0,\ell_1,\ldots,\ell_{k-1}$, and computes $y_j := \sum_{i=0}^{j} \ell_i$ and $Y_j := \mathcal{H}(y_j)$ for $j=0,\ldots,k-1$.
Then, the sender constructs the payload of $P_i$ as follows ($i\geq 1$):
\begin{eqnarray*}
	E_{pk_i} (M_i)  \define E_{pk_i}(P_{i+1}, v_i, (Y_{i-1},Y_{i},\ell_{i}), T_i, E_{pk_{i+1}}(M_{i+1}) )
\end{eqnarray*}
where $E_{pk_i}$ denotes the encryption with the public key of node $P_i$ and $M_{i+1}$ is the payload of node $P_{i+1}$, which is in the same structure as $M_i$. Note that there is always a padding to ensure that all messages have the same length, which we exclude for brevity.  
Here, $T_i = \sum_{j=i+1}^{k-1} TL_j$ is the timelock value and $v_i \define v + \sum_{j=i+1}^{k-1} f_j$ is the amount.
Like in the Lightning, we assume that timelock $TL_i$ and fee $f_i$ values of a channel $\channel_{i}$ as well as the public key of a party $pk_i$ are publicly known.
For the balance security of an honest party $P_i$,  the timelock value $TL_i$ should be greater than or equal to $6\timebroadcast$.
We explain the reasoning of the condition in the atomicity proof.

\paragraph{Locking}
The sender $P_0$ initiates the payment with $P_1$ wrt. the conditional payment $\condpay_{0} = (v_0,Y_0,T_0)$, see step 1 in Fig.~\ref{fig:3hops}.
Note that at this point $P_1$ is the payee, thus she accepts the payment if the channel balance is enough.
If the update is successful, then $P_0$ sends $E_{pk_1} (M_1)$ to $P_1$.
Then, $P_1$ decrypts $M_1$  and checks the payment conditions for the next hop.
After checking the validity of the payment, as described below, $P_1$ initiates the payment with $P_2$ in $\channel_{1}$ (step 2 in Fig.~\ref{fig:3hops}).
If the update is successful, $P_1$ sends $E_{pk_2} (M_2)$ to $P_2$.
The steps of decrypting the message, checking the validity of the payment, and forwarding it to the successor continue until the receiver has been reached.  
If all channels on the path, including the channel of the receiver, agree to lock the payment, then the locking phase is considered successful.
However, if a channel on the path does not accept the payment request, then the locking phase is aborted.  
All the previously locked channels can be unlocked either by agreement of the parties to abort the payment or after the corresponding timelocks expire.

Now, we discuss the channel update requirements for a paying intermediary party $P_{i}$ after updating the channel $\channel_{i-1}$ wrt. $\condpay_{i-1} {=} (v_{i-1},Y_{i-1},T_{i-1})$.
First, $P_{i}$ decrypts $M_{i}$ and obtains $(P_{i+1}, v_i,\allowbreak (Y_{i-1},Y_{i},\ell_{i}), T_i, E_{pk_{i+1}}(M_{i+1})  )$ where $\condpay_{i} {=} ((Y_{i-1},Y_{i},\ell_{i}),T_i)$.
Party $P_i$ accepts to lock the payment $\condpay_{i} {=} (v_{i},Y_{i},T_{i})$ in $\channel_{i+1}$ if:
\begin{itemize}
	\item $P_i$ has enough balance in the channel, at least $v_i$ coins.
	\item The payment value difference between $\condpay_{i-1}$ and $\condpay_{i}$ is at least the expected fee, i.e., $v_{i-1} -v_{i} \geq f_i $.
	\item The timelock difference between $\condpay_{i-1}$ and $\condpay_{i}$ is at least the expected timelock value, i.e., $T_{i-1} -T_{i}\geq TL_i$.
	\item The AMHL condition is correct, i.e., $H(\ell_i) \oplus Y_{i-1} = \mathcal{H}(y_i) = Y_i$, where $\oplus$ denotes the homomorphic operation in the range of $\mathcal{H}$.
\end{itemize}
If all the checks are successful and both parties agree to the update, $P_i$ stores $\ell_{i}$ to be utilized in the payment phase and the channel $\channel_{i-1}$ is updated by the addition of $\condpay_{i} $.
These checks are crucial for the balance security of the honest $P_{i}$.
The timelock check ensures that $P_{i}$ has enough time to react to the protocol regardless of the actions of the other parties; the AMHL check ensures that $P_{i}$ is paid in $\channel_{i-1}$ if she pays in $\channel_{i}$. 
For the receiver $P_k$, since there is no further channel,
Party $P_k$ only checks the timelock in $\condpay_{k-1} \define (v, T_{k-1}, Y_{k-1})$ is greater than or equal to the expected timelock value, i.e., $T_{k-1} \geq TL_k$.
Once the last channel is updated, the locking phase is successful.

\paragraph{Payment}
In this phase, the payment is executed. After the receiver $P_k$ sends $\mathsf{Locked}$ message to the sender $P_0$ (step 4 in Fig.~\ref{fig:3hops}),
 $P_0$ reveals the witness $y_{k-1}$ to $P_k$ (step 5 in Fig.~\ref{fig:3hops}).
The receiver $P_k$ uses the witness $y_{k-1}$ to claim the payment of $\condpay_{k-1}$ from $P_{k-1}$ in channel $\channel_{k-1}$ (step 6 in Fig.~\ref{fig:3hops}). 
If the witness is correct, then $P_{k-1}$ can obtain witness $y_{k-2}$ by using $y_{k-1}$ and already known value $\ell_{k-1}$, i.e., $y_{k-2}=y_{k-1}-\ell_{k-1}$.
Similarly, $P_{k-1}$ can claim the payment of $\condpay_{k-2}$ in channel $\channel_{k-2}$ (step 7 in Fig.~\ref{fig:3hops}). 
The same steps are repeated and the channels are updated in the direction from the receiver to the sender.
If the witness $y_i$ is not sent to $P_i$ before the timeout $T_i$, then the conditional payment $\condpay_{i}$ expires and the parties update their channel $\channel_{i}$ such that $P_{i}$ reclaims the locked amount $v_i$.

In general, upon receiving the witness $y_i$,  an intermediary $P_{i}$,  checks the following conditions:
\begin{itemize}
	\item The witness is correct wrt. the condition, i.e., $\mathcal{H}(y_i) \stackrel{?}{=} Y_i$.
	\item The secret is shared before the timeout, i.e.,  for the time $t_i$ when receiving the secret, $t_i < T_{i}$.
\end{itemize}
If both checks are satisfied, then the parties update their channel $\channel_{i}$ where $P_{i}$ pays $P_{i+1}$ the amount $v_i$.
After receiving the valid witness $y_i$, $P_i$ computes $y_{i-1}{=}y_{i}-\ell_{i}$ and immediately requests the payment from $P_{i-1}$ in the channel $\channel_{i-1}$.

In case of a dispute when one of the parties does not accept to update the channel, the honest party starts the procedure of channel closing via the committee members.
Here, we describe both cases where either payer or payee is honest and the other party does not cooperate with the channel update.
If the payee $P_{i}$ is honest and has shared the witness $y_{i-1}$ on time but the other party $P_{i-1}$ does not update the channel, $P_{i}$  initiates the closure of the channel. 
The procedure works as follows:
\begin{enumerate}
	\item At time $T_{i-1} - \timebroadcast$, if the update is not completed, $P_{i}$ sends the latest state of the channel and the witness $y_{i-1}$ to the committee members $W(P_{i-1},P_i)$.
	\item $W(P_{i-1},P_i)$ validates the state by checking whether the hash of the state is the same as the latest hash (with the signatures of the nodes) they have received. If the check fails, they do not continue. Otherwise, they inform $P_{i-1}$ by sending a message stating the channel closure has been initiated.
	\item Then, $W(P_{i-1},P_i)$ checks if the conditional payment is satisfied, i.e., $y_{i-1}$ satisfies the condition and the time of delivery is before the timelock $T_{i-1}$. If the checks hold, then they approve the payment, and sign the latest state, which includes the payment of $v_{i-1}$, and send their signatures to $P_{i}$. They send $y_{i-1}$ to $P_{i-1}$.
	\item After receiving the signatures of  $W(P_{i-1},P_i)$, $P_i$ signs the new state of the channel and publishes it on the blockchain.
\end{enumerate}
If the payer $P_{i-1}$ is honest and has not received $y_{i-1}$ until the timeout and the other party $P_{i}$ does not update the channel to allow $P_i$ to reclaim her coins, then $P_{i-1}$  initiates the closure of the channel. 
The procedure works as follows:
\begin{enumerate}
	\item At time $T_{i-1}$, if the update is not completed, $P_{i-1}$ sends the latest state of the channel to the committee members $W(P_{i-1},P_i)$.
	\item $W(P_{i-1},P_i)$ validates the state by checking the latest hash they received (as in the previous case). If the checks fail, they do not continue. Otherwise, they inform $P_{i}$.
	\item Then, $W(P_{i-1},P_i)$ checks if the conditional payment has failed, i.e., the timelock $T_{i-1}$ has passed. If the checks hold, then they approve the cancellation of the payment,  sign the latest state, which removes the payment of $v_{i-1}$, and send their signatures to $P_{i-1}$.
	\item After receiving the signatures of $W(P_{i-1},P_i)$, $P_{i-1}$ signs the new state of the channel and publishes it on the blockchain.
\end{enumerate}
Note that if $P_i$ is honest and $P_{i-1}$ has not accepted the channel update, $P_i$ should initiate the channel closure at  $T_{i-1} - \timebroadcast$ as explained in the first case.
Thus, the committee members do not need to validate that $P_{i-1}$ is honestly claiming that she did not receive the witness as they would already have received a message from $P_i$ if $P_{i-1}$ was misbehaving.  
Also, if there are multiple ongoing multi-hop payments, honest parties follow the same procedure for each of them (without sharing the channel state if it is already shared). 
In order to enable that state updates with timelocks about to expire can be sent to the committee despite on-going payments,  parties do not initiate new updates for the respective channel within $4\delta$ of a timelock expiring.  
They further ensure that there are at least $4\delta$ between timelocks for the same channel. 

\subsection{Security and Privacy Analysis}

We show in the following that \sname{} achieves atomicity (def.~\ref{def:atom}), coin availability (def.~\ref{def:availability}), value privacy (def.~\ref{def:v_priv}), endpoint privacy (def.~\ref{def:ep_priv}) and relationship anonymity (def.~\ref{def:r_anony}).
We omit the proof of balance security (def.~\ref{def:bal}),  as atomicity implies balance security.  
Furthermore,  as  it is straightforward that the protocol is successful if all parties follow the protocol and channels have enough balance, we omit the correctness proof.  

To create multi-hop payments and the corresponding payload messages for \sname{}, we utilize two cryptographic operations: hashing $\mathcal{H}(\cdot)$ and encryption $E_{pk}(\cdot)$.
Thus, the security analysis relies on the security of these primitives.
More specifically, we assume that  $\mathcal{H}$ is a cryptographically secure additively homomorphic one-way function and $E_{pk}$ is IND-CCA secure encryption scheme.
 
  \begin{theorem}
  \sname{} provides atomicity.
  \end{theorem}
   \begin{proof}[Proof]
   	Our atomicity proof has three steps: (i) the balance security for an intermediary party $P_i$, (ii) balance security of the sender $P_0$, (iii) atomicity requirement.
   	
   	If the honest party $P_i$ does not pay in $\channel_{i}$, then it does not lose any coins.
   	Thus, we only need to ensure that if $P_i$ pays in $\channel_{i}$, the party receives the payment in $\channel_{i-1}$.
   	In other words, if the channel $\channel_{i}$ is updated with the acceptance of the payment, then there should be sufficient time for the party to use the corresponding witness on the channel $\channel_{i-1}$.
   	
   	If the payment for $\channel_{i}$ is accepted (by parties or the committee), the latest time that $P_{i}$ can receive the witness $y_i$ occurs when the counter party $P_{i+1}$  does not provide $y_i$  before time $T_{i} - \timebroadcast$, but rather shares it with the committee via initiating the channel closure.
   	To have approval of the honest members in $W(P_{i-1},P_i)$, they should receive $y_i$  before $T_{i-1}$. They then send it to $P_{i}$, which is delivered at latest at $T_{i}+\timebroadcast$.
   	After that, an honest $P_{i}$ can compute the witness $y_{i-1}=y_{i}-\ell_{i}$. Then, $P_i$ first attempts to update the channel  $\channel_{i}$ before $T_{i-1} - \timebroadcast$, which takes at most $4 \timebroadcast$ (otherwise $P_i$ sends the channel state and the witness to committee).
   	Thus, for an honest party $P_{i}$ to react on time, the following inequality has to hold: $ T_{i-1} - \timebroadcast - (T_{i}+\timebroadcast )\geq  4 \timebroadcast$.
   	This is satisfied because of the choice of timelocks, i.e., $T_{i-1} -T_{i}\geq TL_i \geq 6\timebroadcast$.  This shows (i). 
   	
   	Now, we show that $P_0$ pays in $\channel_{0}$ only if the receiver $P_k$ is paid in $\channel_{k-1}$.
   	Similar to~\cite{DBLP:conf/ndss/MalavoltaMSKM19}, this is satisfied because of the secure one-wayness property of $\mathcal{H}$.
   	Party $P_0$ only pays if $y_0 = \ell_0$ is provided by $P_1$. The intermediary parties only know the $\ell_i$ values for $i=1,\ldots,k-2$.
   	 The value $y_0$ can only be obtained after $y_{k-1}$ is shared by $P_0$, which happens after the channel  $\channel_{k-1}$ is locked.
   	 In that case, $P_{k}$ can also claim the payment in $\channel_{k-1}$ using $y_{k-1}$. Thus, (ii) holds. 
   	 
   	 Finally, we show atomicity: for an honest sender $P_0$ and honest intermediary $P_{i+1}$, if the condition $\cond_{i}$ is not satisfied, then both conditions $\cond_{\alpha}$ and $\cond_{\beta}$ should not be satisfied for $\alpha<i<\beta$. 
   	 As explained for (ii), in our protocol, the payment conditions are satisfied in the order from the receiver to the sender. 
   	 It is hence sufficient to show that $\cond_{\alpha}$ is not satisfied. 
   	 Note that satisfying the payment condition $\cond_{\alpha}$ implies knowing the preimage of $Y_{\alpha}$, which is $y_{\alpha}= \sum_{i=0}^{\alpha} \ell_i$.
   	 Since $\cond_{i}$ is not satisfied, party $P_{i+1}$ would not know/reveal $y_{i+1}$ or $y_{i}$, and also $\ell_{i+1}$.
   	 Thus the only way to know $y_{\alpha}$ is directly from $Y_{\alpha}$, which contradicts with one-wayness of the hash function $\mathcal{H}$.
  \end{proof}
  
  \begin{theorem}
  \sname{} provides coin availability.
  \end{theorem}
  \begin{proof}[Proof]
  Each conditional payment has a finite timelock. 
  In the honest case, parties release the locked coins before the timelock expires.
  In the malicious case where one of the parties does not accept the update of the channel, the  honest party can initiate the closure of the channel.
  Then, the committee executes the closure of the channel after waiting for the timelocks, which also requires a finite amount of time.
  \end{proof}
  
  \begin{theorem}
\sname{} provides value privacy.
\end{theorem}
  \begin{proof}[Proof]
	We need to show that the adversary $\adv$ cannot win the Value Privacy Game, i.e., cannot guess which of the values $v^{0}$,$v^{1}$ has been used in the multi-hop payment. 
	In \sname{}, payload messages between the parties in the path are encrypted and shared with the onion routing mechanism.
	Thus, an outsider adversary can only obtain the encrypted payload messages. 
	Thereby, if $\adv$  is able to guess the payment value with more than $1/2$ probability, it implies that $\adv$ obtained useful information from the encrypted payload messages.
	However, this contradicts our assumption that the encryption scheme $E_{pk}$ is IND-CCA secure.
\end{proof}

  \begin{theorem}
	\sname{} provides endpoint privacy.
\end{theorem}
\begin{proof}[Proof]
	Here, we show that the adversary $\adv$ cannot win the Endpoint Privacy Game with significantly more than a  probability of $1/2$.
	First, $\adv \define P_i $ chooses two paths of the same length, payment value, and timelocks, and for both paths, the neighbors of $\adv$, $P_{i-1}$ and $P_{i+1}$ are the same.
	For a randomly chosen path, let $M_i$ be the payload received by $\adv$.
	The previous party in the path ($P_{i-1}$) is the same for both paths, thus, in any case, the payload will be received from $P_{i-1}$, and it will include $ (P_{i+1}, v_i, (Y_{i-1},Y_{i},\ell_{i}), \allowbreak T_i, \allowbreak E_{pk_{i+1}}(M_{i+1}))$.
	Since the next party in the path, the payment value, and timelocks are the same, the distinguishing part of the two potential payments would be $ ( (Y_{i-1},Y_{i},\ell_{i}), E_{pk_{i+1}}(M_{i+1}))$.
	The adversary cannot obtain any information from $ E_{pk_{i+1}}(M_{i+1})$ without violating the IND-CCA security assumption of the encryption scheme.
	Finally, we need to show that the tuple $(Y_{i-1},Y_{i},\ell_{i})$ does not reveal any information about the actual sender or receiver pair.
	Note that $\ell_{i}$ values are randomly chosen by the sender of the payment and the $Y_{i}$ values are computed accordingly. 
	Since $\ell_{i}$ values do not contain any information specific to the sender and are chosen randomly, from the perspective of $\adv$, the tuple does not give any information regarding the path.
	Hence, we can conclude that the payload does not give any information about the path other than the neighbors of the adversary. 
\end{proof}

  \begin{theorem}
	\sname{} provides relationship anonymity.
\end{theorem}
\begin{proof}[Proof]
	Different from the endpoint case, in the Relationship Anonymity Game, the adversary is trying to distinguish $b$ from two successful payments.
	Initially, the adversary chooses a path segment $path[1:k-1]$ that is common in both payments.
	Moreover, the adversary can include multiple adversarial intermediary parties in $path[1:k-1]$, so long as there is at least one honest party. 
	
	Let $P_i$ be an honest party and the rest of the intermediaries be adversarial, i.e., $\adv \define path[1:i-1] \cup path[i+1:k-1]$. 
	Note that this is the best case for the adversary (since it includes all but one honest party in the path segment).
	Let $\{M^{0}\}_\adv,\{M^{1}\}_\adv$ be the corresponding payload messages received by $\adv$.
	Since the intermediaries, payment values, and timelocks are the same fro both paths, we can ignore them for distinguishing the paths.
	Thus, the only aspects that can be used for distinguishing the paths are the conditional payment tuples.
	
	The adversary has the tuples $(Y^{b}_{j-1},Y^{b}_{j},\ell^{b}_{j})$ for $b=0,1$ and $j=1,\ldots,i-1,i+1,\ldots,k-1$.
	Trivially,  due to $H$ being homomorphic,  the adversary can link the tuples from the sender to $path[i-1]$.
	Similarly,  $\adv$ can link the tuples from the receiver to  $path[i+1]$.
	Therefore, $\adv$ can win the game (with more than $1/2$ probability) if and only if $\adv$ can link the tuples of $path[i-1]$ and the tuples of $path[i+1]$, thereby connecting the sender and receiver pairs.
	
	Let us denote the tuple of $path[i-1]$ linked to $s^{0}$ with $(Y^{s^{0}}_{i-2}, \allowbreak Y^{s^{0}}_{i-1},\ell^{s^{0}}_{i-1})$ and the one linked to $s^{1}$ with $(Y^{s^{1}}_{i-2},Y^{s^{1}}_{i-1},\ell^{s^{1}}_{i-1})$.
	Similarly,  the tuple of $path[i+1]$ linked to $r^{0}$ is $(Y^{r^{0}}_{i},Y^{r^{0}}_{i+1},\allowbreak \ell^{r^{0}}_{i+1})$ and the one linked to $r^{1}$ is $(Y^{r^{1}}_{i},Y^{r^{1}}_{i+1},\ell^{r^{1}}_{i+1})$.
	Here, if $Y^{s^{0}}_{i-1}$ and $Y^{r^{0}}_{i}$ are linked, i.e., belong to the same payment, then adversary can conclude that $b=0$, otherwise $b=1$.
	If they belong to the same path then, $H(\ell^{b}_i) \oplus Y^{s^{0}}_{i-1} = Y^{r^{0}}_{i}$, otherwise $H(\ell^{b}_i) \oplus Y^{s^{0}}_{i-1} = Y^{r^{1}}_{i}$ where $b\in\{0,1\}$.
	Therefore deriving whether $Y^{s^{0}}_{i-1}$ and $Y^{r^{0}}_{i}$ are linked or not can be reduced to the knowledge of $H(\ell^{b}_i)$ value.
	Since $\ell^{0}_i$ and $\ell^{1}_i$ values are only given to the intermediary party $P_i$, the hash values $H(\ell^{0}_i)$ and $H(\ell^{0}_i)$ that distinguishes the possible paths are not known by $\adv$.
	More specifically, from the adversary's perspective, both $Y^{s^{0}}_{i-1} \oplus Y^{r^{0}}_{i}$ and $Y^{s^{0}}_{i-1} \oplus Y^{r^{1}}_{i}$ values are equally likely candidates for $H(\ell^{b}_i)$ where $b\in\{0,1\}$.
	Under the assumption that the encryption scheme $E_{pk}$ is IND-CCA secure, and $\mathcal{H}$ is a cryptographically secure hash function,
	the adversary cannot guess $b$ value with a probability significantly higher than $1/2$. 
\end{proof}

%% file: 4_protocol_correctcommittees.tex
\section{\aname: Blockchain-Asynchronous Payments on Partially-Synchronous Networks}
\label{sec:async}

This section describes \aname{}, our multi-hop payment protocol that assumes the availability of a BFT committee associated with each channel but does not assume synchronous communications.

\subsection{Model and Overview}

\paragraph{Threat Model}

Committees in \aname{} are assumed to contain at most $f$ faulty members out of $n\ge3f{+}1$ members. 
Channel committee members can be selected randomly from a global blockchain committee using recent  
works~\cite{bortnikov2009brahms, algorand2017,ouroboros2017,david2021vrf},
which guarantee unbiased uniform selection and termination of the selection protocol.  Appx.~\ref{sec:committee} discusses the probability for a committee of size $n$ that is sampled from a global committee of size $N$ with $F$ faulty members to contain less than $f$ faulty members. 

\paragraph{Payment channels}
Committee members maintain a channel's balance, lock the amounts that correspond to on-going payments, and eventually release or confirmed these locked amounts.
To ensure that a payment is eventually processed or rejected,  intermediary parties cannot stop a payment that uses one of their channels, but they are kept informed of their balance.  
To close a channel one of the party has to explicitly request it from the committee members. 
Since the balance of a channel depends on the multi-hop payment that it is involved it,  the committee members then stop accepting new payments on the channel. Once all pending payments have been processed,  the value reported by at least $2f{+}1$ committee members is published on the blockchain as the final channel balance. 

\paragraph{Message format and payment Ids}
Once a channel has been created, like in \sname{}, the payment sender and receiver agree outside of the protocol on the payment. The sender identifies a path to the receiver. 
Messages use the same layered encryption format as in \sname{} with the difference that they do not include timeouts and tuples $(Y_{i-1}, Y_i, \ell_i)$ of hashes and random numbers. 
A committee member of a multi-hop payment in \aname{} trusts subsequent committees on the path to process it.
Upon success, the sender eventually assembles a payment proof made of $f{+}1$ confirmations from its committee. Committees use a total order broadcast to safely process potential concurrent multi-hop payments originating from different payment senders. \aname{} could leverage virtual channels~\cite{DBLP:conf/sp/DziembowskiEFM19} to use only one consensus operation per multi-hop channel creation/deletion and otherwise leverage reliable broadcast for subsequent payments over the multi-hop channel. 

A payment is given a monotonically increasing payment Id by the nodes sending or transmitting it because messages can be reordered by the network and to allow multiple identical payments to be correctly processed. Initially, the payment sender uses its locally maintained payment Id to communicate with its committee. Upon transmitting a payment, a committee replaces the Id it received by an Id it maintains to exchange with the subsequent committee. 

\subsection{Multi-hop protocol}

We present pseudocode of \aname{} in Alg.~\ref{alg:aname}.
For clarity,  we simplify it as follows: we do not show message encryptions or decryptions; the channel balance is only maintained for the payment sender; onion paths are replaced by the full payment path so that a payment can be identified using only the payment Id chosen by the payment sender; payment fees are also ignored. 

\begin{algorithm}[tbp]
\caption{\textit{Simplified \aname{}.}}
\label{alg:aname}
\begin{algorithmic}[1]
\footnotesize 

\State \hspace*{.5em} \underline{For party $P_i$:}

\State \textbf{function} \texttt{initChannel(Amount $v$)}
    \State \hspace*{2em} $myAvailBalance = myBalance = v$
    \State \hspace*{2em} $Id = 0$
    \State

\State \textbf{function} \texttt{pay(Node dest, Amount $v_0$)} \comalgo{Only for $P_0$} \label{l:initSt}
    \State \hspace*{2em} $path = \texttt{findPathTo}(dest)$
    \State \hspace*{2em} $myAvailBalance {\minuseq} v_0$ \label{l:lock}
    \State \hspace*{2em} $\texttt{send}(PAY, Id\texttt{++}, v, path)$ to $W(P_0, P_1)$ \label{l:initEnd} 
\State

\State \textbf{upon rcv} $f{+}1$ $(\textsc{REJECT}, Id,v_i,path)$ from $W(P_i, P_{i+1})$ \textbf{do}
    \State \hspace*{2em} $myAvailBalance {\pluseq} v_i$ 
\State

\State \textbf{upon rcv} $f{+}1$ $(\textsc{SUCCESS}, Id,v_i,path)$ from $W(P_i, P_{i+1})$ \textbf{do}
    \State \hspace*{2em} $myBalance {\minuseq} v_i$ 
\State

\State \hspace*{.5em} \underline{For members of committee $W(P_i, P_{i+1})$:}

\State \textbf{function} \texttt{initChannel(Amount $v$)}
    \State \hspace*{2em} $balance(P_i) = availBalance(P_i) =v$
    \State \hspace*{2em} $nextId(P_{i}) = 0$
\State 

\State \textbf{upon rcv} $(\textsc{PAY}, Id, v_i, path)$ from $P_i$ \textbf{do} \label{l:orderPaySt}
    \State \hspace*{2em} \textbf{if} $Id {\ge} nextId(path[0])$ \textbf{then} \comalgo{discard invalid Ids}
        \State \hspace*{4em} \texttt{totallyOrder}$(\textsc{PAY}, Id, v_i, path)$ \label{l:orderPayEnd}
    \State
    
\State \textbf{upon rcv} $f{+}1$ $(\textsc{PAY}, Id, v_{i}, path)$ from $W(P_{i-1},P_i)$ \textbf{do}  \label{l:orderPaySt2}
    \State \hspace*{2em} \texttt{totallyOrder}$(\textsc{PAY}, Id, v_{i}, path)$
    \State \label{l:orderPayEnd2} 
    
\State \textbf{upon rcv} $f{+}1$ $(\textsc{SUCCESS}, Id, v_{i+1}, path)$ from $W(P_{i+1},P_{i+2})$ \textbf{do}
    \State \hspace*{2em} \texttt{totallyOrder}$(\textsc{SUCCESS}, Id, v_{i+1}, path)$
    \State
    
\State \textbf{upon rcv} $f{+}1$ $(\textsc{REJECT}, Id, v_{i+1}, path)$ from $W(P_{i+1},P_{i+2})$ \textbf{do}
    \State \hspace*{2em} \texttt{totallyOrder}$(\textsc{REJECT}, Id, v_{i+1}, path)$
    \State
    
\State \textbf{upon} $(\textsc{PAY}, Id, v_i, path)$ ordered \textbf{and} $nextId(path[0]) \texttt{==} Id$ \textbf{do}  \label{l:interSt}
    \State \hspace*{2em} $nextId(path[0])\texttt{++}$
    \State \hspace*{2em} \textbf{if} $availBalance(P_i) \ge v_i$ \textbf{then}
        \State \hspace*{4em} $availBalance(P_i) {\minuseq} v_i$
        
        \State \hspace*{4em} \textbf{if} $path[-1] == P_{i+1}$ \textbf{then} \comalgo{last committee} \label{l:recCSt}
            \State \hspace*{6em} $balance(P_i) {\minuseq} v_i$
            \State \hspace*{6em} \textbf{if} $P_i$ has predecessor $P_{i-1}$ in $path$ \textbf{then}
                \State \hspace*{8em} $\texttt{send}(\textsc{SUCCESS}, Id, v_{i}, path)$ to $W(P_{i-1}, P_i)$
            \State \hspace*{6em} \textbf{else}
                \State \hspace*{8em} $\texttt{send}(\textsc{SUCCESS}, Id, v_i, path)$ to $P_i$
            \State \hspace*{6em} Inform $P_i$ and $P_{i+1}$ of new balance 
                \label{l:recCEnd}
        \State \hspace*{4em} \textbf{else} \comalgo{transfer to next committee}
            \State \hspace*{6em} $\texttt{send}(\textsc{PAY}, Id, v_{i+1}, path)$ to $W(P_{i+1}, P_{i+2})$
    \State \hspace*{2em} \textbf{else} \comalgo{insufficient balance}
        \State \hspace*{4em} \textbf{if} $P_i$ has predecessor $P_{i-1}$ in $path$ \textbf{then}
            \State \hspace*{6em} $\texttt{send}(\textsc{REJECT}, Id, v_i, path)$ to $W(P_{i-1}, P_i)$ 
        \State \hspace*{4em} \textbf{else}
            \State \hspace*{6em} $\texttt{send}(\textsc{REJECT}, Id, v_i, path)$ to $P_i$ 
    \State \label{l:interEnd}
    
\State \textbf{upon} $(\textsc{REJECT}, Id, v_i, path)$ ordered \textbf{and} being processed \textbf{do} \label{l:rejectSt}
    \State \hspace*{2em} $availBalance(P_i) {\pluseq} v_i$
    \State \hspace*{2em} \textbf{if} $P_i \texttt{==} path[0]$ \textbf{then}
        \State \hspace*{4em} $\texttt{send}(\textsc{REJECT}, Id, v_i, path)$ to $P_i$
    \State \hspace*{2em} \textbf{else} 
        \State \hspace*{4em} $\texttt{send}(\textsc{REJECT}, Id, v_i, path)$ to $W(P_{i-1}, P_i)$ \label{l:rejectEnd}
    \State
    
\State \textbf{upon} $(\textsc{SUCCESS}, Id, v_i, path)$ ordered \textbf{and} being processed \textbf{do} \label{l:doneISt}
    \State \hspace*{2em} $balance(P_i) \texttt{-=} v_i$
    \State \hspace*{2em} \textbf{if} $P_i$ has predecessor $P_{i-1}$ in $path$ \textbf{then}
        \State \hspace*{4em} $\texttt{send}(\textsc{SUCCESS}, Id, v_i, path)$ to $W(P_{i-1}, P_i)$
    \State \hspace*{2em} inform $P_i$ and $P_{i+1}$ of new balance \label{l:doneIEnd}

\end{algorithmic}
\end{algorithm}

\begin{figure}[t]
	\centering
	\includegraphics[width=\columnwidth]{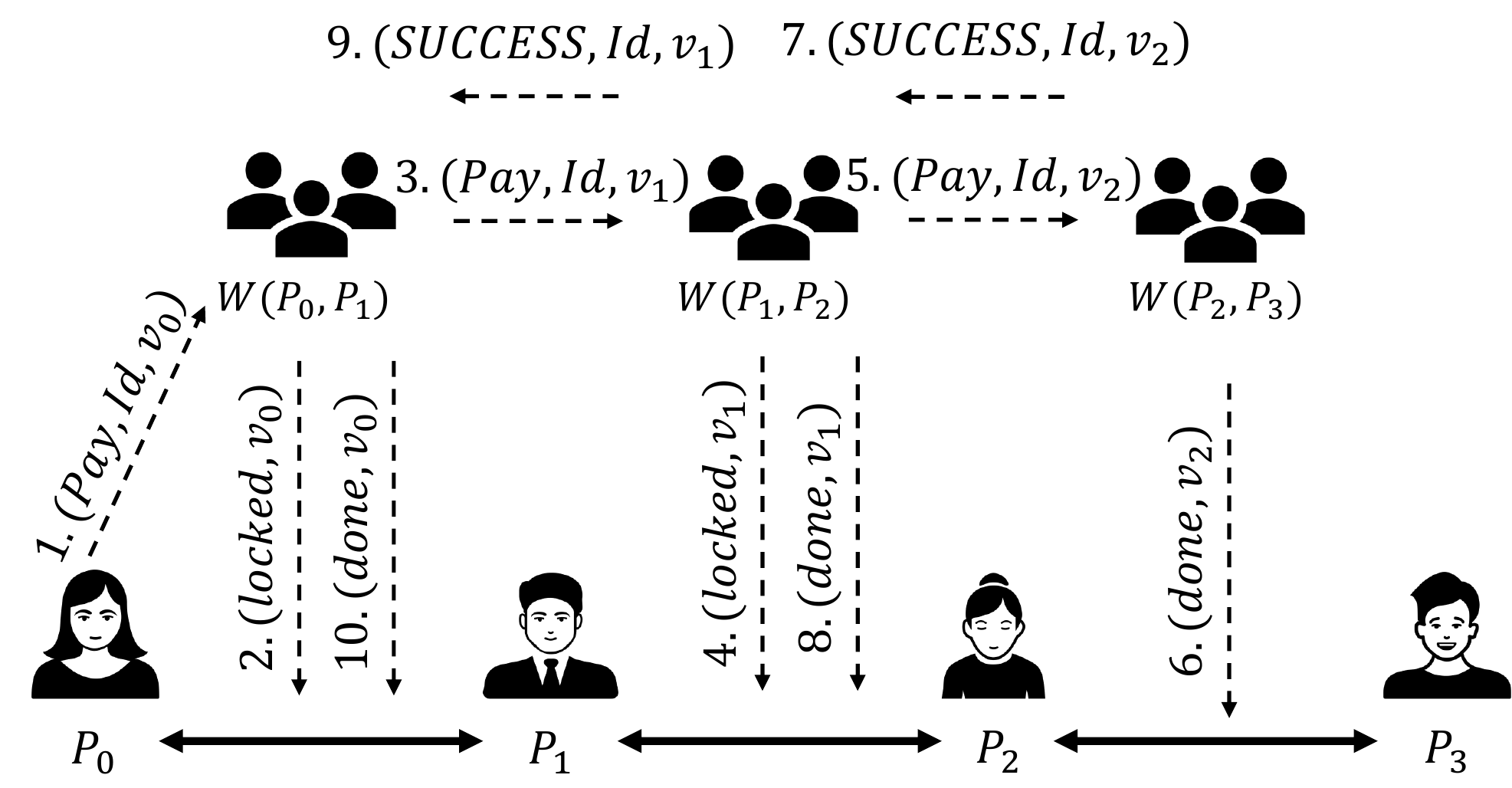}
	\caption{Illustration of the steps of \aname.}
	\label{fig:protocol2}
\end{figure}

To start a payment,  the sender sends a $\textsc{PAY}$ message to the committee of the first channel on the path towards the receiver. 
\aname{} follows a 2-phase approach, like \sname{},
with one phase for locking and one for confirming/revoking a payment. Payments implicitly belong in one of two states:
\begin{enumerate}
\item locked: the payment has been locally registered and transferred to the subsequent committee for processing.
\item done: the payment is completed for this channel. 
\end{enumerate}

Committee members and parties do not limit the time a payment can take. They trust subsequent committees to inform them if a payment is validated or rejected.  The last channel on a payment path does not use the locked state since it can directly approve  or reject a payment. 

After determining the amount and path, the sender $P_0$ informs the committee $W(P_0, P_1)$ of the first channel on the path that they intend to make a payment using a signed message that includes their local monotonically increasing payment Id, the amount, and the subsequent path (Lines~\ref{l:initSt}--\ref{l:initEnd} in Alg.~\ref{alg:aname}, and message 1 in Fig.~\ref{fig:protocol2}).  

The committee $W(P_0, P_1)$ then processes the payment request (Lines~\ref{l:orderPaySt}--\ref{l:orderPayEnd}). More precisely, the payment is added to a pool of payment requests that are then totally ordered by the committee members. Note that payments are not ordered according to any particular criteria so that all payments have an equal chance of being executed.  
After a payment request has been ordered,  the committee checks if i) it has the payment Id the committee is expecting from the payment sender and ii) the channel has enough balance.
If so,  all honest members decrease the available balance of $P_0$ by the amount of the payment. They send a signed message to the next committee informing them that the payment has been marked as locked and include the payment Id, the amount, and the path (Lines~\ref{l:interSt}--\ref{l:interEnd}, msg. 3 in Fig.~\ref{fig:protocol2}). They also inform $P_0$ and $P_1$ that the amount of the payment has been locked on the channel (msg. 2, Fig.~\ref{fig:protocol2}). 

Once a committee $W(P_i,P_{i+1})$ has received at least $f{+}1$ $\textsc{PAY}$ messages from members of the previous committee, the payment is added to the local pool of payment requests of the committee, and is eventually totally ordered  (Lines~\ref{l:orderPaySt2}--\ref{l:orderPayEnd2}). Once the payment has been ordered and is the next payment generated by the payment sender, committee members check whether the channel has enough balance to transfer the amount.  If so,  they lock the amount (in the full protocol, a committee uses a different monotonically increasing payment counter as an Id to transfer a payment to another committee). 
Last, they send a $\textsc{PAY}$ message to the next committee (msgs. 3 and 5 in Fig.~\ref{fig:protocol2}).

To complete our explanation of the first (locking) phase, it remains to consider the case when  $P_{i+1}$ is the receiver.  In this case,  $W(P_i, P_{i+1})$ does not need to wait for a further committee to decide on the payment.  They can simply declare the payment done and inform $P_{i+1}$ as well as a previous committee (if it exists) that the payment is completed (msgs. 6 and 7 in Fig.~\ref{fig:protocol2}). The process is described in Lines~\ref{l:recCSt}--\ref{l:recCEnd}. 

In the second phase, upon receiving $f{+}1$ \textsc{SUCCESS }messages from the subsequent committee indicating that the payment has been completed,  committee members totally order the \textsc{SUCCESS} message, update the balance of the two parties and send a \textsc{SUCCESS} message to the previous committee if it exists (Lines~\ref{l:doneISt}--\ref{l:doneIEnd}, msg. 9 in Fig.~\ref{fig:protocol2}). Upon receiving a \textsc{SUCCESS} message, committees can inform the parties (msgs. 8 and 10, Fig.~\ref{fig:protocol2}).
Rejected payments are also processed with total order and transmitted to the previous committee (Lines~\ref{l:rejectSt}--\ref{l:rejectEnd}). However, if a payment is rejected, only the payment sender is informed since they can then attempt to use a different path in the payment channel network if it exists.  
The other intermediaries are also informed so that they know that the funds are unlocked and can be used for other transactions.  

\subsection{Security and Privacy Analysis}
 
\aname{} does not provide privacy properties since the channels are controlled by BFT committees. In the following, we show that \aname{} satisfies the security properties, namely, atomicity and coin availability.
We omit the proof of correctness property.

\begin{theorem}
\aname{} provides atomicity.
\end{theorem}
\begin{proof}[Proof Sketch] 
If $P_i$ pays in $\channel_{i}$, it means that $2f{+}1$ members of the associated committee voted to totally order the payment and execute it. At least $f{+}1$ correct members of this quorum send a \textsc{SUCCESS} message with the payment Id to the committee members of $\channel_{i-1}$. Upon receiving $f{+}1$ \textsc{SUCCESS} messages, the correct committee members of $\channel_{i-1}$ update $P_i$'s balance. 
Since all committees tolerate $f$ faulty members, the \textsc{SUCCESS} message is sent (starting from the receiver) if all channels have enough balance, otherwise, the locked channels are revoked using \textsc{REJECT} messages.
In other words, $P_0$ only pays after all intermediaries have paid and the receiver $P_k$ has been paid by $P_{k-1}$. Therefore, all channels pay or none do. Furthermore, there is no scenario where $\channel_{i}$ is not paid,  but $\channel_{\alpha}$ and $\channel_{\beta}$ are paid where $\alpha{<}i{<}\beta$.
\end{proof}
 
\begin{theorem}
\aname{} provides coin availability.
\end{theorem}
\begin{proof}[Proof Sketch]
A payment that does not terminate remains in state $\locked$ forever. 
By assumption, messages are eventually delivered, which means that such a payment is stalled by a committee that does not send $f{+}1$ $\textsc{PAY}$ messages to the next committee, or $f{+}1$ $\textsc{SUCCESS}$ or $\textsc{REJECT}$ to the previous committee. 
However, all committees contain at least $2f{+}1$ honest members who upon delivering messages from the previous or next committee react accordingly. Due to the use of total order broadcast, the honest committee members of the committee eventually all agree on the state of a channel as they update it with the same messages and in the same order. Thus, at least $2f{+}1$ members eventually vote on a payment and provide sufficient signatures to allow a payment to eventually terminate.   
\end{proof}

%% file: 5_security.tex
\section{Discussion and Comparison}
\label{sec:discussion}

This section discusses how committee members are incentivized to behave honestly. It then discusses the complexity and latency of our two protocols. Last, it compares their security and privacy properties to those of other payment channel networks.  

\subsection{Incentives}

\begin{table}[t]
\centering
\caption{Network latencies and expected message complexities of a $k$-hop payment with committees of size $n$. }
\begin{tabular}{p{2.4cm}|p{3.2cm}|p{1.7cm}}
  & Msg complexity & Latency \\
 \hline
 \sname & $8nk + 3k + 2$ & $(8k+2)\delta$ \\
 \hline
 \aname & $(2k{-}1)(2n{+}C_{msg}(n){+}1)$ & $2kC_{lat}(n){+}2\delta$ \\
 \hline
 \aname (w.~\cite{castro1999practical}) & $\Theta(4kn^2)$ & $\Theta(6k\delta)$\\
 \hline
 \aname (w.~\cite{yin2019hotstuff}) & $\Theta(18kn)$ & $\Theta(16k\delta)$\\
\end{tabular}
\label{tab:2}
\end{table}

We have so far considered honest and malicious parties/committee members.  
In contrast, 
Brick has three incentives to ensure that rational parties follow the protocol. 
 First, committee members are incentivized to sign updates.  To do so,  members are paid an \emph{update fee} by the channel parties upfront to sign a channel update and a channel party waits for $2f {+1}$ signatures to confirm the update. It may be argued that a committee member could simply collect the update fee and not sign the update. 
 However, channel parties avoid such committee members for future updates and thus decrease their long-term expected rewards,  meaning that rational parties indeed sign the update.
Second,  it prevents rational committee members from misbehaving,  e.g., by submitting an invalid state during closure. 
During channel opening,  committee members deposit a \emph{collateral} (at least equal to $C/f$, where $C$ is the channel capacity) that remains locked until the channel is closed. If the channel closes in the freshest state, then the collateral is returned to committee members. However, if a channel party provides a proof that a committee member misbehaved, then the collateral is paid to the channel party.  
A malicious channel party may bribe rational committee members to close in an invalid state. The bribe, however, needs to be larger than the collateral locked by the committee member. 
The cost of such a bribe is proven to exceed the capacity of the channel and hence is not beneficial for a channel party.  Thus,  rational channel parties do not execute bribing attacks and rational committee members only submit correct states.  
Third, it incentivizes committee members to sign the closing transaction promptly by paying a \emph{closing fee} to
the first $2f{+}1$ signing committee members.  

These incentives can be directly included in our protocols so that they would assume rational parties and committee members instead of correct ones. 
It is also necessary to incentivize intermediaries to forward a payment, even if they have to pay update fees to their committee members.  
We can leverage the fee mechanism used in Lightning, where intermediaries receive a fee upon successful completion of a payment, but need to ensure that the average fee of the intermediary exceeds the fee paid for the channel update.  
Let $p_S$ be the approximated success probability of payments (based on previous payments) and $f_i$ be the fee of an intermediary $P_i$.  The expected earned fee is then $f_{i}{\cdot} p_S$. 
Note that multi-hop payments require two updates and for each update the channel party has to pay the update fee $f_{CM}$ to all $3 f{+}1$ committee members to receive $2 f{+}1$ signatures in the worst case.
For rational intermediaries to forward a payment, we need $f_{i} {\cdot} p_S {>} 2 (3f+1)f_{CM}$.

\subsection{Message Complexity and Latency}
\label{sec:msg_complexity}

Table~\ref{tab:2} summarizes the analytical network latency and message complexity of \sname{} and \aname{}. 
The complexities of \aname{} are evaluated under the assumption that the system is in the synchronous phase of a partially synchronous network, i.e., during the phase in which progress can be made.  
$C_{lat}(n)$ and $C_{msg}(n)$ denote the expected network latency and msg complexity, respectively, of the underlying total order broadcast for a committee of size $n$. Potential realizations are PBFT~\cite{castro1999practical} and HotStuff~\cite{yin2019hotstuff}.
Assuming synchronous networks allow \sname{} to have lower message complexity and latency than \aname{}. 
The performance of \aname{} heavily depends on the total order broadcast protocol used. With PBFT~\cite{castro1999practical}, the bandwidth usage is higher than with HotStuff~\cite{yin2019hotstuff}. Inversely, the expected latency of \aname{} is lower with PBFT than with HotStuff.

\begin{table*}[t]
	\centering
	\caption{Comparison of Payment Channel (Networks).  BS\&C\&CA is balance security, correctness and coin availability, R. Anon. refers to relationship anonymity, Par. Sync. is partial synchrony, and L. and G, resp. denote local and global.}
	\begin{tabular}{c|c|c|c|c|c|c|c|c|c|c}
		Ref. & Protocol &Multi-hop & \multicolumn{2}{c|}{ Synchrony}  &  \multicolumn{3}{c|}{Security}  &  \multicolumn{3}{c}{Privacy} \\\
		&   & Payment  & Blockchain & Network  & Secured  by & BS\&C\&CA  & Atomicity & Value & Endpoint & R. Anon. \\
		\hline
		\cite{poon2016bitcoin}	& Lightning & \cmark & \multicolumn{2}{c|}{Synchronous} & HTLC & \cmark & \xmark  & \cmark & \cmark & \xmark \\
		\cite{DBLP:conf/ndss/MalavoltaMSKM19}& AMHL & \cmark & \multicolumn{2}{c|}{Synchronous}   & AMHL & \cmark & \cmark  & \cmark & \cmark & \cmark  \\
		\cite{DBLP:conf/sp/DziembowskiEFM19}& Perun & \cmark & \multicolumn{2}{c|}{Synchronous}  & Smart Contract & \cmark  & \cmark & \cmark & \xmark & \xmark \\
		\cite{DBLP:conf/uss/AumayrMKM21} & Blitz & \cmark & \multicolumn{2}{c|}{Synchronous}  &Pay-or-revoke & \cmark & \cmark  & \cmark & \cmark & \xmark \\
		\cite{DBLP:conf/fc/0001BBKM19} & Sprites & \cmark & \multicolumn{2}{c|}{Synchronous}   & Smart Contract & \cmark & \cmark  & \cmark & \xmark & \xmark  \\
		\hline
		\hline
		\cite{lind2018teechain}	& Teechain & \cmark &  Async. & Par. Sync.  & TEE & \cmark & \cmark  & \cmark & \cmark & \cmark \\
		\cite{avarikioti2019brick}	& Brick & \xmark & \multicolumn{2}{c|}{Asynchronous}  & L. Committee  & NA & NA  & NA & NA & NA \\
		\hline
		\hline
		\S~\ref{sec:sync} &	\sname & \cmark & Async. & Sync. & L. Com. \& AMHL & \cmark & \cmark  & \cmark & \cmark & \cmark \\
		\S~\ref{sec:async} &	\aname & \cmark & Async. & Par. Sync. & G. Committee & \cmark & \cmark   & \xmark & \xmark & \xmark \\
		\hline
	\end{tabular}
	\label{tab:1}
\end{table*}

\subsection{Comparison of Payment Channel Networks}

Several protocols have considered reducing the locked time of coins in Lightning~\cite{DBLP:conf/fc/0001BBKM19}, improve the payment path privacy~\cite{DBLP:conf/ndss/MalavoltaMSKM19}, minimize the communication rounds~\cite{DBLP:conf/uss/AumayrMKM21}, and increase the efficiency of multi-hop payments~\cite{DBLP:conf/sp/DziembowskiEFM19}.
However, the security of these protocols relies on blockchain synchrony.
Teechain~\cite{lind2018teechain} is the only multi-hop payment protocol that supports partial synchrony thanks to TEEs. 
   
Table~\ref{tab:1} summarizes our analysis and considers multi-hop payment support, synchrony assumptions (for the blockchain and the network), as well as the security and privacy properties we defined in Section~\ref{sec:securepcn}.
These properties apply to multi-hop payments, but we also compare to Brick~\cite{avarikioti2019brick} and use NA (not applicable) for it. 
For security, we combine the basic security properties, balance security, correctness, and coin availability, that are satisfied by all of the protocols.  

With regard to atomicity,  Blitz~\cite{DBLP:conf/uss/AumayrMKM21} uses a common revocation transaction to prevent the wormhole attack where some intermediaries are skipped during a successful payment.  
Similarly,  the global manager in Sprites~\cite{DBLP:conf/fc/0001BBKM19} ensures that every intermediary receives the preimage and updates its channel.  
Perun~\cite{DBLP:conf/sp/DziembowskiEFM19} guarantees atomicity because payments involve the endpoints of virtual channels and intermediaries do not explicitly participate.
\sname{} adapts mechanisms from AMHL~\cite{DBLP:conf/ndss/MalavoltaMSKM19}, whereas \aname{} assumes BFT committees. One might consider providing incentives against the attack as an alternative protection against rational behaviors.  

We consider value privacy, endpoint privacy, and relationship anonymity in the absence of side channel attacks. 
Our second protocol, \aname{}, does not provide any privacy since the committees control the channels. They see the payment values and the committee of the receiver is aware of its identity.
The authors of Sprites state that their model does not provide privacy guarantees, though they do not reveal the payment value and hence achieve value privacy.  
Thereby, our analysis focuses on the rest of the protocols.
Value privacy is about hiding the payment amount from outsiders, and it is satisfied by all protocols except \aname{}.
Endpoint privacy can be achieved by onion routing,  which hides the sender and receiver from the intermediaries.  Lightning, AMHL,  and Blitz all use onion routing.  
Since the virtual channel construction in Perun reveals the endpoints to the intermediary parties, endpoint privacy and thereby relationship anonymity are not satisfied.
Relationship anonymity simply implies that non-connected intermediary parties should not know if they are part of the same payment or not.
Lightning and Blitz do not satisfy this property because of the common hash and revocation transaction shared by all channels in the payment path.
\sname and AMHL protocols satisfy this property using the additive hash construction, which was initially proposed in~\cite{malavolta2017concurrency}.
Finally, Teechain~\cite{lind2018teechain} also satisfies the privacy properties under the assumption that TEEs do not collude and leak any information about the payments, e.g., due to side channel attacks~\cite{brasser2017software, wang2017leaky}.

%% file: 6_conclusion.tex
\section{Conclusion}
\label{sec:conclusion}

In this paper, we have demonstrated that it is impossible to design a multi-hop payment protocol assuming both network asynchrony and faulty channels.   
We then showed that network synchrony or correct channels allow the design of a multi-hop payment protocol by presenting one protocol for each case. 
We detailed two different committee-based multi-hop payment protocols that assume synchronous communications and possibly faulty channels, or asynchronous communication and correct channels, respectively.  
Our protocols do not require blockchain synchrony to solve disputes among parties. 
\sname{} tolerates faulty committee and assumes synchronous communications.  
\aname{} relies on BFT committees that totally order payment requests, which can be implemented using several recent consensus protocols depending on the synchrony model one assumes, and trust each other to eventually process a payment.
Because it has stronger synchrony assumptions, \sname{} generates fewer messages and has lower latency than \aname{}. In future work, we will consider using virtual channels in \aname{}, which could allow committees to leverage reliable and consistent broadcast in multi-hop payments instead of total order broadcast.  We will also work on adding privacy guarantees to \aname{}. Furthermore,  we want to design a version of \aname{} using an accountable BFT protocol~\cite{civit2021polygraph} to deal with committees with more than $f$ fulty nodes.

%% file: 7_committee.tex
\section{Committee formation in \aname{}}
\label{sec:committee}

\begin{figure}[b]
    \centering
    \includegraphics[width=1\columnwidth]{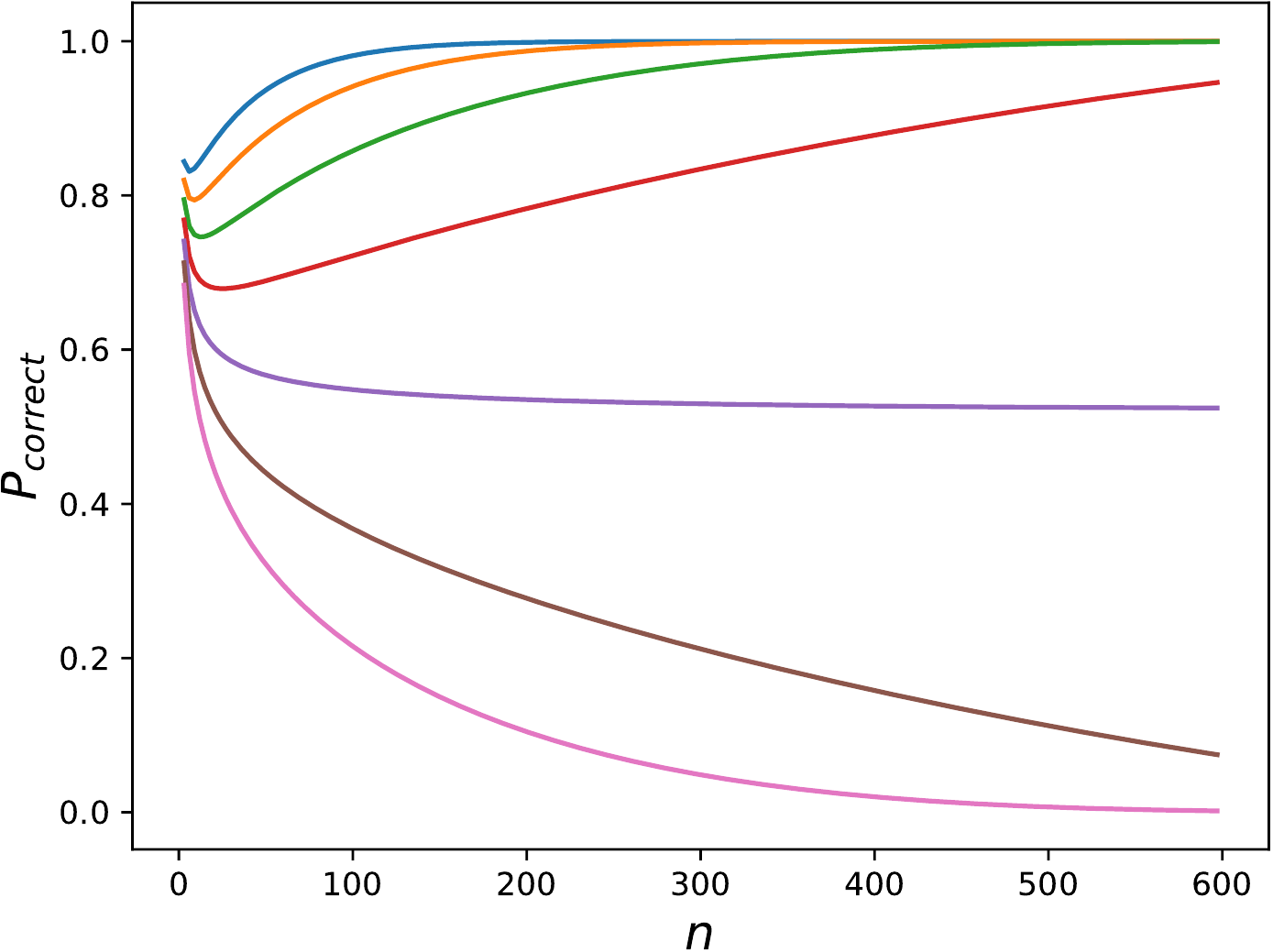}
    \caption{Probability of choosing a committee $P_{correct}$ when $N{=}1200$. We plot the variation of $P_{correct}$ with respect to $n$ for each $F$ in the set $[300,325,350,375,400,425,450]$.}
    \label{fig:committeeselection}
\end{figure}

We assume that there is a global committee with $N$ members from which a smaller committee of size $n$ is sampled during the channel opening to be the channel committee. 
The channel committee members can be selected randomly from the global committee based on any of a number of recent peer selection methods~\cite{bortnikov2009brahms, algorand2017,ouroboros2017,david2021vrf},
which guarantee unbiased uniform selection and termination of the selection protocol. Assuming that the global committee has $F$ faulty members, we can compute the probability $P_{correct}$ for a correct committee ($f\leq\lfloor n/3\rfloor)$) to be selected. The probability $p_f$ that the selected committee has $f$ faulty members follows a hyper-geometric distribution~\cite{abdelatif2019sharding} and is equal to $\frac{{F\choose f}\cdot {{N-F}\choose {n-f}}}{{N\choose n}}$. It follows that $P_{correct} = \sum_{f=1}^{\lfloor n/3 \rfloor} p_f$. While selecting a bigger channel committee does increase the likelihood of the committee to be correct, a bigger channel committee also implies higher latency and bandwidth use. 

Fig.~\ref{fig:committeeselection} indicates the probability for a sampled channel committee to contain more than $f$ faulty members.
For this figure, we assume a global committee of size $N=1200$, and vary the value of $F \in [300,325,350,375,400,425,450]$.  One can observe that $P_{correct}$ increases with $n$ when $F<N/3$ and increases with $n$ when $F>N/3$.  For $N/3$, it converges.    
When $F=N/4$, we see that $P_{correct}$ is nearly $1$ for all values of $n$. As the value of $F$ moves closer to $N/3$, a higher value of $n$ is required for $P_{correct}$ to approach $1$.  

The probabilities to sample a faulty committee of size 300 with F being 300 and 325 are almost equal to 1, namely $0.999$ and $0.998$, respectively. 
We therefore select committees of size $300$ for our performance evaluation, and assume a global committee that contains less than one-fourth of faulty members. 

Large BFT committees have limited performance when used in blockchains to process the transactions of all users. In our protocols,  a committee handles the transactions of a single channel and is therefore not expected to provide the same performance under high workload as a blockchain consensus algorithm.  
There is a small probability for a sampled committee to be faulty (i.e., contain more than $f$ faulty members).